\documentclass[lettersize]{IEEEtran}
\usepackage{amsmath,amsfonts}
\IEEEoverridecommandlockouts
\normalsize

\usepackage{cite}
\usepackage{url}
\usepackage{hyperref}
\hypersetup{colorlinks=true, linkcolor=red, citecolor=blue, urlcolor=blue} 
\usepackage{mathrsfs}
\usepackage{amssymb}
\usepackage{amsmath}
\usepackage{array}
\usepackage{amsthm}
\allowdisplaybreaks 
\usepackage{autobreak}
\usepackage{mathptmx}
\usepackage{mathtools, cuted}

\usepackage{graphicx}
\usepackage{subcaption} 
\usepackage{bbm}
\usepackage{booktabs}
\usepackage{multirow}
\usepackage{makecell}
\usepackage[table,xcdraw,dvipsnames]{xcolor}
\usepackage{colortbl}
\usepackage{enumitem}

\usepackage{algorithm}
\usepackage{algorithmic}

\newtheorem{remark}{Remark}
\newtheorem{proposition}{Proposition}

\usepackage{acronym}
\acrodef{ml}[ML]{machine learning}
\acrodef{drl}[DRL]{deep reinforcement learning}
\acrodef{fl}[FL]{federated learning}
\acrodef{hfl}[HFL]{hierarchical federated learning}
\acrodef{rnn}[RNN]{recurrent neural network}
\acrodef{fcn}[FCN]{fully connected neural network}
\acrodef{lstm}[LSTM]{long short-term memory}
\acrodef{gru}[GRU]{gated recurrent unit}
\acrodef{fedavg}[FedAvg]{federated averaging}
\acrodef{ue}[UE]{user equipment}
\acrodef{bs}[BS]{base station}
\acrodef{isp}[ISP]{(wireless) internet service provider}
\acrodef{csp}[CSP]{content service provider}
\acrodef{es}[ES]{edge server}
\acrodef{sgd}[SGD]{stochastic gradient descent}
\acrodef{chr}[CHR]{cache hit ratio}
\acrodef{wgan}[WGAN]{Wasserstein generative adversarial network}
\acrodef{mdp}[MDP]{Markov decision process}
\acrodef{milp}[MILP]{mixed-interger linear programming}
\acrodef{ilp}[ILP]{interger linear programming}
\acrodef{csi}[CSI]{channel state information}
\acrodef{uav}[UAV]{unmanned aerial vehicles}
\acrodef{lru}[LRU]{least recently used}
\acrodef{lfu}[LFU]{least frequently used}
\acrodef{rl}[RL]{reinforcement learning}
\acrodef{cs}[CS]{cloud server}

\newcommand{\bblue}{\textcolor{black}}

\begin{document}
%
\title{
Revenue Optimization in Wireless Video Caching Networks: A Privacy-Preserving Two-Stage Solution}

\author{Yijing Zhang, Md Ferdous Pervej, \textit{Member, IEEE}, and Andreas F. Molisch, \textit{Fellow, IEEE}
\thanks{A preliminary version of this paper is currently under review at the 2025 IEEE Military Communications Conference (MILCOM) \cite{zhang2025revenue}.}
\thanks{This work was funded by NSF under the NSF-IITP project $2152646$.}
\thanks{Y. Zhang and A. F. Molisch are with the Ming Hsieh Department of Electrical and Computer Engineering, University of Southern California, Los Angeles, CA 90089 USA (e-mail: yijingz3@usc.edu; molisch@usc.edu).}
\thanks{M. F. Pervej was with the Ming Hsieh Department of Electrical and Computer Engineering, University of Southern California, Los Angeles, CA 90089 USA. He is now with the Department of Electrical and Computer Engineering, Utah State University, Logan, UT 84322 USA (e-mail: ferdous.pervej@usu.edu).}}

\maketitle

\begin{abstract}

Video caching can significantly improve delivery efficiency and enhance quality of video streaming, which constitutes the majority of wireless communication traffic. Due to limited cache size, 
caching strategies must be designed to adapt to and dynamic user demand in order to maximize system revenue.
The system revenue depends on the benefits of delivering the requested videos and costs for (a) transporting the files to the users and (b) cache replacement. Since the cache content at any point in time impacts the replacement costs in the future, demand predictions over multiple cache placement slots become an important prerequisite for efficient cache planning.
Motivated by this, we introduce a novel two-stage privacy-preserving solution for revenue optimization in wireless video caching networks. 
First, we train a \emph{Transformer} using \bblue{privacy-preserving} \ac{fl} to predict multi-slot future demands. 
Given that prediction results are never entirely accurate, especially for longer horizons, we further combine global content popularity with per-user prediction results to estimate the content demand distribution. 
Then, in the second stage, we leverage these estimation results to find caching strategies that maximize the long-term system revenue. This latter problem takes on the form of a multi-stage knapsack problem, which we then transform to a integer linear program. 
Our extensive simulation results demonstrate that (i) our \ac{fl} solution delivers nearly identical performance to that of the ideal centralized solution and outperforms other existing caching methods, and (ii) our novel revenue \bblue{optimization approach} provides deeper system performance \bblue{insights} than traditional \ac{chr}-based \bblue{optimization} approaches.
\end{abstract}


\IEEEpeerreviewmaketitle

\acresetall
\section{Introduction}
\subsection{Background and Motivation} 

The growing demands for video streaming create increasing pressure on wireless networks \cite{PopcachingSurvey}, especially on backhaul links.
As repeated requests for popular content drive up traffic, ensuring low latency and stable service becomes more challenging. 
The backhaul load can be significantly reduced by edge caching, i.e., storing the most frequently requested files at the network edge (and/or the user's local devices) \cite{golrezaei2013femtocaching}.
However, the storage size of an edge server, which can be embedded in the \ac{bs} (or local devices), is typically much smaller than the total amount of library content. Thus, efficient caching necessitates knowledge of the future video demands.


Since efficient caching relies on predicting users' future requests, accurate prediction of future demands is one of the key requirements for cache design. However, it is challenging because a user's \emph{future} content requests can deviate significantly from both their own past usage patterns and from the global popularity. 
Since the classical caching methods (e.g., \ac{lru}, \ac{lfu}, etc.) focus on using past information to make cache decisions, they fail to capture the change in content popularity over time \cite{pervej2024resource}.
Therefore, an efficient caching solution requires accurate demand predictions and shall shuffle the cached content considering cost versus predicted benefits.
\Ac{ml} can greatly enhance the efficiency of cache design by capturing user-specific preferences and forecasting their future demands. 
In particular, \emph{foundation models} such as the Transformer \cite{attentionAllyouNeed} are a promising tool to accurately predict these demands.


In wireless video caching networks, users as well as commercial entities (e.g., \ac{isp} and \ac{csp}) want to preserve the privacy or confidentiality of the data. 
Users are reluctant to share their personal request information due to privacy concerns. 
In addition, when a user requests content from the \ac{csp} through their serving \ac{isp}, confidentiality issues arise between the \ac{csp} and \ac{isp}. 
Since these two entities are often business competitors, the \ac{csp} prefers to keep user-specific request details confidential. 
Similarly, the \ac{isp} is unwilling to disclose the location of user requests.
Given these constraints, prediction of demand at a \ac{bs} cannot rely on centrally aggregated knowledge by the \ac{isp} or the \ac{csp}, making centralized \ac{ml} 
algorithm unsuitable. 
Therefore, a privacy-preserving learning solution such as \ac{fl} is required for demand prediction that facilitates the caching strategy. 
Note that \ac{fl} enables distributed model training at the wireless edge, which also ensures that information remains protected between different parties while allowing effective demand prediction.


Based on the predicted content requests, an efficient cache placement policy needs to be developed. 
Given that the total number of available content files far exceeds the storage capacity at the edge, it is crucial to design an appropriate objective function to determine which content should be cached within the limited storage.
While the \ac{chr} has long been used as the target function for cache design, it only reflects the user satisfaction but ignores the cost of caching, and thus fails to provide a comprehensive evaluation. 
In practice, the objective function should balance two key trade-offs: (i) the long-term benefits of maintaining or updating the cached content, which helps reduce backhaul costs associated with retrieving requested files, and (ii) the expenses incurred in refreshing the cache.
Therefore, an efficient caching strategy should be guided by an objective function that incorporates these trade-offs from a revenue optimization perspective, ensuring both cost-effectiveness and improved system performance.

\subsection{Related Work}


\Ac{ml} is extensively used in the literature to capture the dynamic changes in content requests \bblue{and to predict future demands} \cite{MLEdgecachingPowerAllo, MLCodedcachingUnkonw, MLCloudAidedHierarchical,lekharu2024collaborativeRL, DRLjointcachingRoutingAI, DRLThreeTier}. 
These studies can be broadly categorized into (1) supervised learning and (2) \ac{rl}.
For supervised \ac{ml}, recent works \cite{MLEdgecachingPowerAllo, MLCodedcachingUnkonw, MLCloudAidedHierarchical} mainly used \ac{lstm} based model to predict future content demands. 
\Ac{drl} is also widely used \cite{lekharu2024collaborativeRL, DRLjointcachingRoutingAI, DRLThreeTier} to capture long-term popularity changes for caching policy design. 
However, these works fundamentally consider that the \ac{ml} model is trained in a centralized setting, i.e., assuming training data are centrally available. 
However, the training data are generated by the local users, who may not be  willing to share their raw data due to the above-discussed privacy concerns.


Many recent studies \cite{pervej2024resource,selfDrivingLSTMHFL,yu2018federated,Delay_tolarantFLcaching,NonIIDFLcaching,InEdgeAI,DRLforIoT,MDPF_DRL,HCP,CPPPP} used \ac{fl} for different cache placement applications to satisfy privacy constraints in wireless networks. 
\cite{pervej2024resource} proposed a collaborative privacy-preserving \ac{hfl} algorithm to predict content demands in the next timeslot. 
\cite{selfDrivingLSTMHFL} devised an \ac{hfl} algorithm for content popularity prediction using attention scores as input to the \ac{ml} model.
\cite{yu2018federated} used the \ac{fedavg} algorithm \cite{McMahanFL} to train an autoencoder model, which extracts latent user data and generates content recommendation lists for caching decisions. 
\cite{Delay_tolarantFLcaching} also used a stacked autoencoder for global content popularity prediction. 
It uses a priority mechanism during training, which avoids inaccurate learning due to missing model parameters and reduces feedback delay. 
\cite{NonIIDFLcaching} introduced an adaptive \ac{fl}–based proactive content caching algorithm that uses \ac{drl} to select participating clients and to determine the number of local training rounds.
Similarly, \cite{InEdgeAI} integrated \ac{drl} with \ac{fl} to optimize computing, caching, and communication resources.
It models popularity‐aware cache replacement at each edge node as a \ac{mdp}, with local \ac{drl} agents deciding whether to cache incoming content and using the resulting reward feedback for federated model training.
\cite{DRLforIoT} developed a federated \ac{drl}-based cooperative edge caching framework, which models the content replacement problem as a \ac{mdp} and train \ac{drl} agents to solve the problem.
\cite{MDPF_DRL} also formulates the caching policy as a \ac{mdp} and uses Dueling Deep Q Network to solve the \ac{mdp}.
\cite{HCP} considered a \ac{uav}-assisted caching network, where the cache placement is predicted by jointly considering traffic distribution, \ac{ue} mobility, and localized content popularity using \ac{fl}. 
\cite{CPPPP} proposed a content popularity prediction strategy using \ac{wgan}, which were trained in a privacy-preserving distributed manner using.
It uses \ac{wgan} to generate high-quality fake samples for predicting content popularity, thus better securing privacy and achieving a high \ac{chr}.

Most of the above \ac{fl}-based caching methods primarily focus on predicting future user requests or content popularity. 
Once these predictions are obtained, a caching policy needs to be designed using an appropriate utility function. 
Recent works have considered different utility functions to optimize cache placement. 
\Ac{chr} is a widely used utility function and has long been used for designing content placement strategies. 
For example, \cite{JavedankheradCHR2022} proposed a low-complexity weighted vertex graph-coloring algorithm that captures the characteristics of mobile users and efficiently places popular files at \acp{bs} to maximize the \ac{chr}.
However, \ac{chr} is a simplistic utility function and may not always be sufficient. 
Some recent studies have focused on minimizing system costs or maximizing overall benefits \cite{Multicell-Coordinated, jointassortpaper, CostMinFuyaru}.
\cite{Multicell-Coordinated} formulated a collaborative caching strategy aimed at minimizing the total cost incurred by content providers to the network operator.
Similarly, 
\cite{jointassortpaper} further extended the cost minimization approach to a revenue maximization problem through the joint design of personalized assortment decisions and cache planning in wireless content caching networks. 
Besides, \cite{CostMinFuyaru} 
focused on minimizing system costs, which comprises the (a) caching, (b) retrieval, and (c) update costs, by optimizing the caching decisions.

\subsection{Research Gaps and Our Contributions}
Among the above studies, \cite{MLEdgecachingPowerAllo, MLCodedcachingUnkonw, MLCloudAidedHierarchical,lekharu2024collaborativeRL, DRLjointcachingRoutingAI, DRLThreeTier} adopt centralized training approaches, which violate privacy concerns. 
Although some studies \cite{selfDrivingLSTMHFL,yu2018federated,Delay_tolarantFLcaching,NonIIDFLcaching,InEdgeAI,DRLforIoT,MDPF_DRL,HCP,CPPPP} proposed caching methods that aim to protect user privacy, they fail to address collaboration among the three key entities involved, namely, the \ac{ue}, \ac{isp}, and \ac{csp}. Even though some studies such as \cite{pervej2024resource} take different parties into consideration, they do not propose caching methods. Additionally, \cite{CostMinFuyaru, jointassortpaper, Multicell-Coordinated} exhibit limitations in designing objective functions: \cite{Multicell-Coordinated} focuses on file delivery between different \acp{bs}, \cite{jointassortpaper} does not consider revenue optimization in relation to short-term changes in future popularity, and \cite{CostMinFuyaru} assumes known user preference distributions, which are usually unknown in real applications. 
Moreover, a revenue model should incorporate past and future caching decisions, a consideration that is absent in these works. 
Our paper aims to address these limitations.

In this paper, we propose a two-stage proactive caching solution, where a \emph{Transformer} \cite{attentionAllyouNeed} is trained with a privacy-preserving mechanisms to predict users' future requests across multiple time slots.
These predictions are then utilized to optimize cache placement decisions to maximize long-term revenue. 
Our key contributions are summarized as follows:
\begin{itemize}
\item As the video caching networks benefit from delivering the requested content to the users and have costs for (a) content placement, (b) delivery from \ac{es}' cache, and (c) extraction from the \ac{csp}'s cloud server, we define a new revenue function and formulate a revenue optimization problem considering multiple future cache placement slots, each consisting of many content request mini-slots. 
\item Since the actual future request is unknown beforehand, we use a privacy-preserving \ac{fl} solution \cite{pervej2024resource} with a Transformer 
to predict users' future content requests for multiple future slots. 
\item However, since the predictions for long sequences are not always accurate, we model the actual content requests as a random variable and model its distribution based on a combination of the predicted requests from the trained Transformer, local content popularity, and prediction accuracy.
\item We then use these estimated future demands and formulate the revenue optimization problem, which we then transform into an  \ac{ilp} problem, which can be efficiently solved using existing tools.
\item 
Our simulation results demonstrate that the proposed solution outperforms other counterparts in terms of \ac{chr} and revenue, with revenue offering a more comprehensive reflection of overall system performance than \ac{chr}. 
\end{itemize}


The remainder of the paper is organized as follows. 
Section \ref{sec:sys_model} introduces the video caching system model. 
In Section \ref{sec:rev_opt}, we formulate the multi-slot revenue optimization problem, followed by the problem analysis and discussions of challenge points. 
Then, Section \ref{sec:two_stage_solution} presents our two-stage solution. 
Section \ref{sec:sim_results} provides extensive simulation results and discussions. Finally, we conclude the paper in Section \ref{sec:conclusion}. 
Some necessary notations are summarized in Table \ref{table:FLCL}.

\begin{table}[h!]\small
\centering
\caption{List of Notations}
\begin{tabular}{ |c|c| } 
 \hline
    \bf Parameter & \textbf{Definitions}  \\ 
    \hline
 $u, \mathcal{U}, U$ & User $u$, all user set, number of users  \\ 
 \hline
 $f, \mathcal{F}, F$ & Content $f$, all content set, number of files  \\ 
 \hline
 $\mathfrak{g}, \mathcal{G}, G$ &  Genre $\mathfrak{g}$, all genre set, number of genres \\ 
 \hline
 $P_{\mathfrak{g},f}$ &  Content $f$'s popularity in genre $\mathfrak{g}$ \\ 
 \hline
 $p_{u,\mathfrak{g}}$ &  User $u$'s genre preference \\ 
 \hline
 $g_{u,f}$ &  User $u$'s local content popularity \\ 
 \hline
 $S, B$ &  Storage capacity of the \ac{es}; uniform file size\\ 
 \hline
 $\tau, t$ &  Cache placement slot; mini-slot \\ 
 \hline
 $n$ &  Number of mini-slots in the cache placement slot \\ 
 \hline
 $i_{u,f}^t, \mathbf{I}_{u}^t$ &  \makecell[c]{Binary indicator that defines wether $u$ request \\content $f$ at mini-slot $t$; vector notation of \\user $u$'s content request at $t$} \\ 
 \hline
 $d_f^\tau$ &  \makecell[c]{Binary indicator that defines\\ wether content $f$ stored at slot $\tau$} \\ 
 \hline
 $c_{\mathrm{plc}}$ &  Cache placement cost \\ 
 \hline
 $c_{\mathrm{bs-ue}}$ &  Cost for content delivery from \ac{bs} to \ac{ue} \\ 
 \hline
 $c_{\mathrm{cl-bs}}$ & \makecell[c]{Cost of retrieving content \\from the cloud server to \ac{bs}} \\ 
 \hline
 $\beta$ &  Benefit of successful content delivery \\ 
 \hline
 $R^{\tau}$ & System revenue at slot $\tau$ \\ 
 \hline
 $R_{\mathrm{req}}^{\tau}$ &  Content request revenue at slot $\tau$ \\ 
 \hline
 $R_{\mathrm{plc}}^{\tau}$ &  Placement revenue at slot $\tau$ \\ 
 \hline
 $K,\tilde{K}$ &  \makecell[c]{Number of future slots after the upcoming slot; \\caching effect of further $\tilde{K}$ } \\ 
 \hline
 $k$ &  $k^{\mathrm{th}}$ caching slot \\ 
 \hline
 $\gamma$ & Decaying factor that discount effect of future slots \\ 
 \hline
 $\Theta,\Theta^*$ &  \ac{ml} model; well-trained global model \\ 
 \hline
 $r,r_g,\kappa$ & \makecell[c]{Global training round; \\total global training rounds; \\local training round} \\ 
 \hline
 $\eta, l_u$ &  Learning rate; loss function \\ 
 \hline
 $N$ & Number of content requests for input feature \\ 
 \hline
 $\mathbf{x}_u,\mathbf{\hat{y}}_u$ & Input feature; predicted output \\ 
 \hline
 $R_{\mathrm{total}}$ & Overall system revenue \\ 
 \hline
 $a,\hat{a}$ & \makecell[c]{Probability that actual content request takes \\on value based on prediction; indicator} \\ 
 \hline
 $a_{u,f}^t$ & \makecell[c]{Prediction accuracy of user $u$ \\request content $f$ at mini-slot $t$} \\ 
 \hline
 $K'$ & Total cache placement slots in the validation set \\ 
 \hline
 $\mathbf{S}_{\mathbf{I}}$ & \makecell[c]{Set of all possible values \\of actual content request $\mathbf{I}_u^t$} \\ 
 \hline
 $z_{f}^{\tau+k}$ & Product of adjacent cache decisions $d_f^{\tau+k} d_f^{\tau+k-1}$ \\ 
 \hline
 $L,M$ & \makecell[c]{Previous $L$ requests \\for generating following $M$ requests} \\ 
 \hline
 $\mathcal{F}_{u,L}$ & \makecell[c]{Previously requested $L$ content set }\\ 
 \hline
 $f_l$ & Content ID at slot $l$ in $\mathcal{F}_{u,L}$ \\ 
 \hline
 $\pmb{\phi}_{f_l}$ & Feature set of the $f_l^{\mathrm{th}}$ content \\ 
 \hline
 $\epsilon, E$ & \makecell[c]{$\epsilon^{\mathrm{th}}$ day; \\Total days of content request model per user} \\ 
 \hline
 $Q$ & Content requests per day \\ 
 \hline

\end{tabular}
\label{table:FLCL}
\end{table}

\section{System Model}
\label{sec:sys_model}
\begin{figure}
\centering
\includegraphics[scale = 0.27]{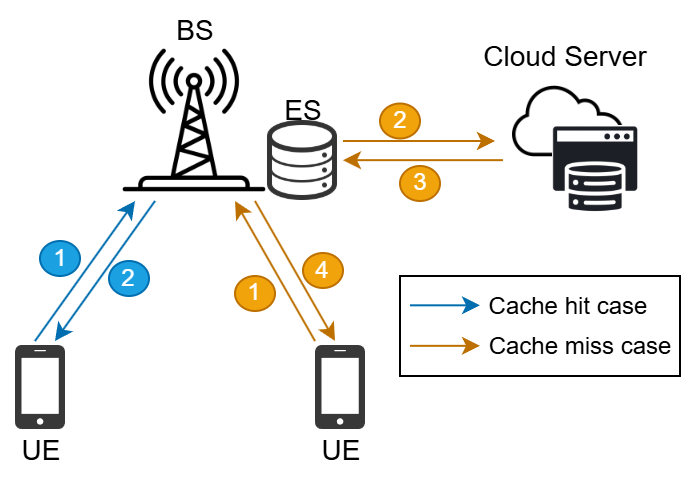}
\caption{Cache placement system model}
\label{fig_systemodel}
\end{figure}

\begin{figure*}[t]
\centering
\includegraphics[width=0.85\textwidth]{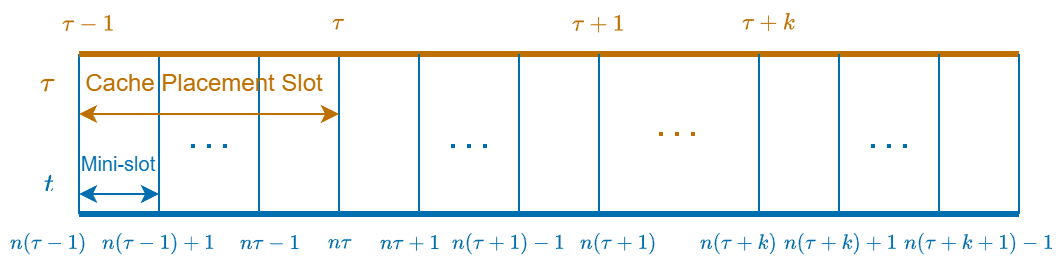}
\caption{User request and cache placement time slots}
\label{cachingReqSlots}
\end{figure*}

\subsection{Video Caching Network Model}

We consider a wireless video caching network comprising a single \ac{bs}, multiple \acp{ue}, and a \ac{cs}, as illustrated in Fig. \ref{fig_systemodel}. 
The sets of \acp{ue} and content are denoted by 
$ \mathcal{U} = \{u\}_{u=0}^{U-1} $ and $ \mathcal{F} = \{f\}_{f=0}^{F-1} $, respectively. 
Each file is assumed to have a uniform size of 
$B$ bits. 
Without loss of generality, we assume the contents are grouped by their genres.
Denote the genre set by $ \mathcal{G} = \{\mathfrak{g}\}_{\mathfrak{g}=0}^{G-1} $. 
Besides, content in each genre has popularity $P_{\mathfrak{g},f}$ and $\sum \nolimits_{f\in \mathfrak{g}} P_{\mathfrak{g},f}=1$. 
The users have their own genre preferences.
Denote the preference of the $u^{\mathrm{th}}$ user by $p_{u,\mathfrak{g}}$, where $\sum \nolimits_{\mathfrak{g}=0}^{G-1} p_{u,\mathfrak{g}}=1$.
Every \ac{ue} maintains its local content request history, from which it derives a personalized content popularity metric, $g_{u,f}$. 
However, these historical data are not shared with other \acp{ue} or the \ac{bs} due to privacy concerns.
The \ac{bs} and cloud server are controlled by the \ac{isp} and \ac{csp}, respectively, operating as independent entities under different business organizations. 
Each entity keeps its operational data confidential from the other. 
To enable privacy-preserving collaboration, the \ac{csp} can deploy an \ac{es} at the \ac{bs}, following the approach in \cite{pervej2024resource}, to facilitate 
privacy-preserving demand predictions, which then guides the cache planning. 
The storage capacity for the cache at the \ac{es} is denoted by 
$S$, where $S < (F \times B)$, where for convenience the file size $B$ is assumed to be identical for all files.
Additionally, we assume that content updates in the \ac{es} cache occur at predefined intervals, namely the start of what is referred to as {\em cache placement} slots, whose duration is denoted by $\tau$. 
However, users can request content from the \ac{csp} at the start of discrete mini-slot $t$, of which there are $n$ within each cache placement slot, as illustrated in Fig. \ref{cachingReqSlots}. 
This cache update strategy is widely adopted in practical wireless edge caching networks \cite{pervej2024efficient}, where the frequency of cache updates typically depends on network traffic conditions (e.g., peak hours) and application requirements \cite{golrezaei2013femtocaching}.

In this paper, we assume that the downlink communication between the \ac{bs} and \acp{ue} follows is error-free. This is a reasonable assumption in 4G or 5G wireless systems, enabled by the use of strong error correction coding together with selection of \bblue{a} proper modulation and coding scheme based on \ac{csi} knowledge at the \ac{bs}; residual errors are detected and corrected via automatic repeat request \cite{molisch2023wireless}.  
As such, we assume the physical-layer transmissions within each mini-slot are reliable and keep our primary focus on the cache planning.


To formally define the content request process, we introduce a binary variable $i_{u,f}^t\in\{0,1\}$ which takes the value of 1 if user $u$ requests content $f$ during mini-slot $t$, and 0 otherwise. 
To conveniently represent all content requests for a given user in a mini-slot, we also use the shorthand vector notation
$\mathbf{I}_{u}^t = (i_{u,0}^t, i_{u,1}^t, \dots, i_{u,F-1}^t) \in \mathbb{R}^F$ that contains all binary content request variables.
Besides, we assume that each user requests only a single video content within a given mini-slot, i.e., $\sum_{ f =0 }^{F-1} i_{u,f}^t=1$, $\forall u \in [0,U-1]$.
Denote the cache placement decision by a binary variable
$d_f^\tau\in\{0,1\} $, which takes the value of $1$ if content $f$ is stored in the \ac{es}' cache during cache placement slot $\tau$ and takes the value of $0$ otherwise.

Caching content at the \ac{es} incurs an additional placement cost, denoted by $c_{\mathrm{plc}}$.
If a requested file is available in the \ac{es}' cache, the \ac{bs} and \ac{es} can collaborate under {\em service-level agreements} to deliver the content locally to the user \cite{pervej2024resource}. 
This scenario is referred to as a {\em cache hit}, and the cost for local content delivery is denoted by $c_{\mathrm{bs-ue}}$.
If the requested content is not available in the \ac{es}' cache, it must be retrieved from the \ac{csp}'s cloud server, which constitutes a {\em cache miss event}. 
The total cost for delivering the requested content in this case consists of both the cloud-to-\ac{bs} retrieval cost, denoted by $c_{\mathrm{cl-bs}}$, and the subsequent delivery cost to the user, i.e., $c_{\mathrm{bs-ue}}$. 
While both cache placement and cloud-based content retrieval involve transmission over the backhaul link, $c_{\mathrm{cl-bs}}$ is generally higher than cache placement cost $c_{\mathrm{plc}}$, i.e., $c_{\mathrm{cl-bs}}>c_{\mathrm{plc}} $.
This is due to the fact that cache placement typically occurs during off-peak hours when network traffic is lower.
Finally, regardless of whether a cache {\em hit} or {\em miss} occurs, the \emph{benefit} of successfully delivering the requested content to the user is denoted by $\beta$, which represents, for instance, the rental fee a user pays for accessing the video.

\section{Revenue Optimization Cache Placement}
\label{sec:rev_opt}

In order to calculate the revenue, we need to consider the costs $c_{\mathrm{plc}}$, $c_{\mathrm{bs-ue}}$, and $c_{\mathrm{cl-bs}}$, as well as the benefits of delivering the requested content. 
Intuitively, prefetching the contents that are likely to be requested for an extended period is advantageous, as the cache placement cost $c_{\mathrm{plc}}$ is typically lower than the backhaul extraction cost $c_{\mathrm{cl-bs}}$ and is incurred only once - there are no additional cache placement costs when a cached file is delivered to the user in the subsequent content request slots.
However, the popularity of video files fluctuates over time at both global and regional levels. 
As a result, the relevance of cached files may diminish, leading to higher probabilities of cache miss events and, hence, increased backhaul costs for serving those requests. 
Therefore, periodic cache re-placement is necessary to maintain efficiency and make the best out of finite cache storage.
Moreover, since the revenue is typically calculated over a long horizon, the cache placement decisions should not be made in isolation: one must consider what has been cached in the past, as well as future demands. 
When updating the cache, minimizing the number of file replacements is desirable to reduce unnecessary costs. 
Therefore, optimizing system revenue by strategically considering the impact of cache placement across multiple cache placement slots is crucial.



Consider the situation at a cache placement slot $\tau \geq 0$, which contains $n$ mini-slots, i.e., the slots where users place their content requests to the \ac{csp}. 
The instantaneous revenue in this particular cache placement slot $\tau$ is
\begin{align}
\label{eq:one_slot_rev}
R^{\tau} 
&\coloneqq R_{\mathrm{benefit}}^{\tau} - R_{\mathrm{delivery\_cost}}^{\tau} - R_{\mathrm{placement\_cost}}^{\tau} \nonumber\\ 
&= \underbrace{\sum_{t = n\tau}^{n(\tau+1)-1} \sum_{f =0}^{F-1} \sum_{u =0}^{U-1} i_{u,f}^t \cdot \beta}_{R_\mathrm{benefit}^\tau} - \nonumber\\
&\quad \underbrace{\sum_{t = n\tau}^{n(\tau+1)-1} \sum_{f =0}^{F-1} \sum_{u =0}^{U-1} i_{u,f}^t \cdot \left[c_{\mathrm{bs-ue}} + (1-d_f^{\tau}) c_{\mathrm{cl-bs}}\right]}_{R_{\mathrm{delivery\_cost}}^\tau} -\nonumber\\
&\quad \underbrace{ c_{\mathrm{plc}} \sum_{f=0}^{F-1} d_f^{\tau} (1-d_f^{\tau-1})}_{R_{\mathrm{placement\_cost}}^{\tau}}, 
\end{align}
where the first term is the benefit of delivering the requested files, the second term is the cost of delivering the requested file to the user, and the last term is the cost of placing the files into the cache.

From a system optimization perspective, the goal is to optimize the cache decisions to maximize system revenue over a long period, which consists of infinitely many cache placement slots. 
Therefore, the ideal revenue of a video caching wireless network is defined as 
\begin{align}
\label{eq:total_revenue}
    R_{\mathrm{system\_revenue}} \coloneqq \lim_{\mathrm{T} \to \infty} \frac{1}{T} \left[ \sum_{\tau=0}^\mathrm{T-1} R^\tau \right]. 
\end{align}

\begin{remark}
Unless future content requests in all cache placement slots are known beforehand, (\ref{eq:total_revenue}) cannot be calculated directly.  
In practical wireless video caching networks, we can neither consider infinitely many cache placement slots to calculate the revenue nor assume that the users will share what content they will request in the far future.
As such, we need to find a way to predict the requests and an objective function to plan our caching policy that maximizes this revenue approximately.  
\end{remark}

Intuitively, (\ref{eq:one_slot_rev}) and (\ref{eq:total_revenue}) suggest that the cache decisions need to be made jointly for all $\mathrm{T}$ caching slots since the placement cost depends on the sequential cache placement decisions. 
Since we do not know the future, and demand prediction becomes increasingly unreliable for longer timescales, we may only be able to reasonably predict the demands for a short period. 
Nonetheless, we need to design an objective function that considers the benefits that storing a file in the current caching slot also over future caching slots.
As such, when we are in cache placement slot $\tau$, we calculate the revenue over some arbitrary $K+\tilde{K}$ caching slots as
\begin{equation}
\label{eq:rev_function}
\begin{aligned}
    & R_{\mathrm{total}}^\tau 
    =\sum_{k=0}^{K+\tilde{K}-1} \gamma^k R^{\tau+k} \\ 
    & = \sum_{k=0}^{K+\tilde{K}-1} \gamma^k \bigg[\sum_{t = n(\tau+k)}^{n(\tau+k+1)-1} \sum_{f =0}^{F-1} \sum_{u =0}^{U-1} i_{u,f}^t( \beta - c_{\mathrm{bs-ue}}  - \\
    &(1-d_f^{\tau+k})c_{\mathrm{cl-bs}}) \bigg] - \sum_{k=0}^{K+\tilde{K}-1} \gamma^k \cdot c_{\mathrm{plc}} \sum_{f=0}^{F-1} d_f^{\tau+k} (1-d_f^{\tau+k-1}),
\end{aligned}
\end{equation}
where 
$0 < \gamma \leq 1$ is a decaying factor that assigns less weight \bblue{to} future cache placement slots \bblue{and is a hyperparameter}.
This formulation assumes that changes of the cache placement will occur in $K$ future slots; while the cache placement beyond that point is assumed to be static, the benefit of the placement is computed for an additional $\Tilde{K}$ slots, and thus, in (\ref{eq:rev_function}),  calculate the revenue over $K+\Tilde{K}$ caching slots. 
Therefore, we want to find the cache decisions for all placement slots, i.e., $\{\{d_f^{\tau+k}\}_{f=0}^{F-1}\}_{k=0}^{K-1}$, to maximize system revenue.
The cache planning optimization problem is thus formulated as 
\begin{subequations}
\label{optproblem}
\begin{align}
&\underset{\{\{d_f^{\tau+k}\}_{f=0}^{F-1} \}_{k=0}^{K-1}} {\mathrm{maximize}} \qquad  R_{\mathrm{total}}^\tau  \tag{\ref{optproblem}} \\
&~~\mathrm{subject ~to} \quad C1:\quad  B\sum_{ f=0 }^{F-1} d_f^{\tau+k} \leq S, k \in [0,K-1], \\ 
&\qquad C2: \quad d_f^{\tau+k}\in\{0,1\},f \in [0,F-1], k \in [0,K-1],
\end{align}
\end{subequations}
where constraints $C1$ and $C2$ are because of the limited cache size at \ac{es} and the binary cache placement decision variables, respectively.

Now, it is important to recall that our goal is to approximately maximize (\ref{eq:total_revenue}). 
While one may stick to the decisions $\{\{d_f^{\tau+k}\}_{f=0}^{F-1} \}_{k=0}^{K-1}$ for all $K$ caching slots obtained from (\ref{optproblem}), this may not be the best strategy due to the following reasons.
At a caching slot $\tau$, we only need to place the content for that particular slot. 
(\ref{eq:total_revenue}) can guide us to look ahead to compute the impact of this decision, i.e., $\{d_f^{\tau+k}\}_{f=0}^{F-1}$, in the future $\tilde{\tau} > \tau$ slots, and thus optimize the decision under the currently available knowledge. However,  
the decision at the next caching slot $\tilde{\tau}+1$ should utilize the full then-available historical information, including the demands during caching slot $\tau$, to compute an improved caching decision.
As such, it is meaningful to solve (\ref{optproblem}) in a sequential manner.
In other words, we decide the cached content in every caching slot $\tau$ by looking at $K$ future slots.

The optimization problem derived in Sec. III.A  
is a multistage knapsack problem with the decision variables over $K$ slots from $k=0$ to $K-1$, and thus, is NP-hard \cite{bampis2022multistage}.


\begin{remark}
The optimization problem  (\ref{optproblem}) is challenging to solve directly for multiple reasons. 
Firstly, cached content should be determined at the beginning of $\tau$, while the revenue relies on future content requests $i_{u,f}^t$, which is unknown. 
Secondly, the last term in (\ref{eq:rev_function}), i.e., $\gamma^k \cdot c_{\mathrm{plc}} \sum\nolimits_{f=0}^{F-1} d_f^{\tau+k} (1-d_f^{\tau+k-1})$, indicates that a previously cached content yields cost savings in the next caching slot.
Intuitively, this sequential coupling of the cache planning and placement cost (saving) needs to be considered to capture the long-term caching benefits.
\end{remark}


Due to the challenges mentioned above, we propose a two-stage solution.
Firstly, we train a Transformer \cite{attentionAllyouNeed} to predict the future content requests using a privacy-preserving \ac{fl} algorithm. 
Then, we show an intuitive way to utilize the prediction results from the Transformer to make caching decisions. 

\section{Problem Transformation and Two-Stage Solution}
\label{sec:two_stage_solution}

The details of the proposed two-stage solution are summarized below.

\subsection{Stage 1: \ac{fl}-aided Demand Predictions}

We assume that each user leverages their historical content requests $\mathbf{I}_{u}^t$ to prepare their respective training dataset, denoted by $\mathcal{D}_u = \{\mathbf{x}_u^m,\mathbf{y}_u^m\}_{m=0}^{D_u-1}$, where ${D_u}$ represents the total number of training samples, $\mathbf{x}_u^m$ is the feature set and $\mathbf{y}_u^m$ is the corresponding label. 
We want to train a Transformer parameterized by weight vector $\Theta$. 
More specifically, our goal is to use $\Theta$ to predict users' content requests for all mini-slots $t \in [n\tau, n(\tau+K)-1]$. 
However, since training data only belongs to the user and is private, we use the \ac{fedavg} \cite{McMahanFL} algorithm to train $\Theta$ while preserving data privacy. 
The training process is summarized below.

\subsubsection{Offline Model Training Using FedAvg\cite{McMahanFL}} 
At the start of each global training round $r$, which refers to a communication cycle where the server coordinates model updates across clients, 
the \ac{es} distributes the global model $\Theta^r$ to all users in $\mathcal{U}$. 
Upon receiving the model, each user updates their local model $\Theta_u$ as $\Theta_u^0 \gets \Theta^r$ and performs $\kappa$ mini-batch \ac{sgd} updates as
\begin{equation}
\Theta_u^{\kappa} = \Theta_u^0 - \eta \sum\limits_{\tau=0}^{\kappa-1} \nabla l \left(\Theta_u^\tau | (\zeta \sim \mathcal{D}_u) \right)
\end{equation}
where $\eta$ represents the learning rate, $\nabla$ is gradient operator, $l$ denotes the loss function, and $\mathcal{D}_u$ is the prepared training dataset. 
After completing the local updates, users upload their trained models to the \ac{es}.
The \ac{es} then aggregates the received models as
\begin{equation}
\Theta^{r+1} = \frac{1}{U}\sum_{u=0}^{U-1} \Theta_u^\kappa. 
\end{equation}
Note that we use equal aggregation weights, i.e., $1/U$, since we assume all users have the same number of training data samples. 
The \ac{es} then shares the updated global model with all users, who repeat the same steps. 
This iterative process continues for $r_g$ global training rounds.

\subsubsection{Predictions Using Trained $\Theta^*$}
\label{sec:predion_proc}
Once the \ac{es} has the trained $\Theta^*$, it broadcasts the model to all users. 
The users then leverage the trained $\Theta^{*}$ to predict their future content requests\footnote{While users do not share their actual content requests, this prediction is made as an ML task.}. 
Since our goal is to forecast content requests for the next $nK$ mini-slots starting at time $t$, the input feature is derived from the past $N$ mini-slot content requests, given by $\mathbf{x}_u = (\mathbf{I}_{u}^{t-N}, \mathbf{I}_{u}^{t-N+1}, \cdots, \mathbf{I}_{u}^{t-1}) \in \mathbb{R}^{N\times F}$.
Denote the corresponding predicted output for the next $nK$ mini-slots by $\mathbf{\hat{y}}_u = (\mathbf{\hat{I}}_{u}^{t}, \mathbf{\hat{I}}_{u}^{t+1}, \cdots, \mathbf{\hat{I}}_{u}^{t+nK-1}) \in \mathbb{R}^{nK\times F}$.
Note that each predicted  $\hat{i}_{u,f}^t \in [0,1]$ in $\mathbf{\hat{I}}_{u}^{t}$ represents the {\em probability} that user 
$u$ will request content $f$ during mini-slot $t$.
All users then share their predicted content requests to the \ac{es}, which processes these results to solve the optimization problem, which are presented in the sequel.

Unfortunately, the predictions are not perfect, which is a quite well-known effect in \ac{ml}.
Typically, the prediction accuracies for time-series data decrease as we move farthest in the prediction time slots.
In our case, since (\ref{optproblem}) relies on the actual content request $i_{u,f}^t$, wrong predictions of $i_{u,f}^t$ directly affect caching decisions, potentially leading to suboptimal cache planning. 
Thus, relying solely on the \ac{fl}-based prediction and replacing $i_{u,f}^t$ by the prediction $\hat{i}_{u,f}^t$ from the trained $\Theta^{*}$ may not be sufficient.

\subsection{Stage 2: Transformer's Prediction Assisted Revenue Optimization}
\label{sec:request_estimation}

To address the prediction inaccuracies, we model the future requests, $i_{u,f}^t$, for all $t \in [n\tau, n(\tau+K)-1]$, as random variable.
Then, the expected achievable revenue is calculated as
\begin{align}
\label{eq:expectedRevEqn0}
&\mathbb{E}[R_{\mathrm{total}}^\tau] 
= \mathbb{E} \bigg[\sum\limits_{k=0}^{K-1} \gamma^k \sum\limits_{t = n(\tau+k)}^{n(\tau+k+1)-1} \sum\limits_{f =0}^{F-1} \sum\limits_{u=0}^{U-1} i_{u,f}^t \big( \beta - c_{\mathrm{bs-ue}} \nonumber \\
&- (1-d_f^{\tau+k})c_{\mathrm{cl-bs}}\big) - \sum\limits_{k=0}^{K-1} \gamma^k \cdot c_{\mathrm{plc}} \sum\limits_{f=0}^{F-1} d_f^{\tau+k} (1-d_f^{\tau+k-1})\bigg] \nonumber\\
& =\sum\limits_{k=0}^{K-1} \gamma^k \sum\limits_{t = n(\tau+k)}^{n(\tau+k+1)-1} \sum\limits_{f=0}^{F-1} \sum\limits_{u=0}^{U-1} \mathbb{E} \left[i_{u,f}^t\right] \times \big( \beta - c_{\mathrm{bs-ue}} \nonumber\\ 
&- (1-d_f^{\tau+k})c_{\mathrm{cl-bs}}\big) - \sum\limits_{k=0}^{K-1} \gamma^k \cdot c_{\mathrm{plc}} \sum\limits_{f=0}^{F-1} d_f^{\tau+k} (1-d_f^{\tau+k-1}),
\end{align}
where $\mathbb{E}[\cdot]$ is the expectation operator that depends on the randomness of the actual content request $i_{u,f}^t$. 

In order to calculate the $\mathbb{E}[i_{u,f}^t]$, we now focus on modeling the random variable $i_{u,f}^t$. Note $i_{u,f}^t$ serves as an indicator at each position of $\mathbf{I}_{u}^t$, to indicate which file is selected.
We assume that $i_{u,f}^t$ follows the following definition.
\begin{align}
    \mathrm{P}[i_{u,f}^t=1] \coloneqq
    \begin{cases}
        \hat{i}_{u,f}^t, & \text{with  probability}~ \bar{a}_{u,f}^{t}, \\
        g_{u,f}, & \text{with probability}~ (1 - \bar{a}_{u,f}^{t}), 
    \end{cases},
\end{align}
where $\bar{a}_{u,f}^{t}$ is an {\em unknown} probability and $g_{u,f} \coloneqq \frac{\sum
_{t \in \mathcal{T}_{\mathrm{his}} } i_{u,f}^t}{\sum_{t \in \mathcal{T}_{\mathrm{his}}} \sum_{f =0}^{F-1}i_{u,f}^t}$, where $\mathcal{T}_{\mathrm{his}}$ denote the historical mini-slot set, is the local content popularity. 
In order to estimate the unknown probability $\bar{a}_{u,f}^{t}$, we use the prediction accuracy from the Transformer $a_{u,f}^{t}$, which is calculated for each \emph{mini-slot} $t$ within \emph{cache placement slot} $\tau$ using the user's validation dataset as 
\begin{equation}
\label{eq:accuracy}
    a_{u,f}^{t} \coloneqq \frac { \sum _{\tau'=\tau}^{\tau+K'} \delta \left( f - \mathrm{argmax} [\hat{\mathbf{I}}_{u}^t] \right) \times \delta \left( f - \mathrm{argmax} [\mathbf{I}_{u}^t] \right) } {\sum_{\tau'=\tau}^{\tau+K'} \delta \left( f - \mathrm{argmax} [\mathbf{I}_{u}^t] \right)}, \forall t \in \tau'
\end{equation}
where $K'$ denotes the total \emph{cache placement slots} in the validation set, and the summation over $t$ is carried out across the corresponding \emph{mini-slot} within each $\tau$ interval. 
Besides, $\delta(\cdot)$ is the delta function.

We have the following proposition.

\begin{proposition}
\label{prop}
Based on our random request model, the prediction accuracy $a_{u,f}^t$ and predicted request $\hat{i}_{u,f}^t$ from the Transformer, the expectation of actual content request is written as
\begin{equation}
\label{eq:estimation_true}
\mathbb{E}[i_{u,f}^t] = \hat{i}_{u,f}^t a_{u,f}^t + g_{u,f} (1-a_{u,f}^t),
\end{equation}
where $\hat{i}_{u,f}^t$ is the $f^{\mathrm{th}}$ entry of the vector $\hat{\mathbf{I}}_u^t$.
\end{proposition}
\begin{proof}
Under condition of transformer prediction $\mathbf{\hat{I}}_u^t$, the probability of actual request $\mathbf{I}_u^t$ is
\begin{equation}
\label{eq:probability}
\begin{split}
\mathrm{P}[\mathbf{I}_u^t | \mathbf{\hat{I}}_u^t] &= \mathrm{P}[\mathbf{I}_u^t | \mathbf{\hat{I}}_u^t \cap \{ \hat{a}=1\}]\mathrm{P}[\mathbf{\hat{I}}_u^t \cap \{ \hat{a}=1\}]
\\&\ \ + \mathrm{P}[\mathbf{I}_u^t | \mathbf{\hat{I}}_u^t \cap \{ \hat{a}=0\}]\mathrm{P}[\mathbf{\hat{I}}_u^t \cap \{ \hat{a}=0\}]
\\& = \mathbf{\hat{I}}_u^t \cdot\mathbf{\bar{a}}_u^t + \mathbf{g}_{u} \cdot (1-\mathbf{\bar{a}}_u^t)
\end{split}
\end{equation}
where $\mathbf{\bar{a}}_u^t = (\bar{a}_{u,0}^t, \bar{a}_{u,1}^t, \dots, \bar{a}_{u,F-1}^t) \in \mathbb{R}^F$ and $\hat{a}$ is indicator function that takes the value of $1$ when the request is based on the Transformer prediction and takes the value of $0$ when it is based on the local popularity.
We denote the set of all possible values of actual content request $\mathbf{I}_u^t$ as a set $\mathbf{S}_{\mathbf{I}}$ with $F$ one-hot-encoded vectors indicating $F$ possible values of the content request indicator function $\mathbf{I}_u^t$. 
The expectation of the actual requests can be written as
\begin{equation}
\begin{split}
\mathbb{E}[\mathbf{I}_u^t | \mathbf{\hat{I}}_u^t] &= \sum_{\mathbf{I}_u^t \in \mathbf{S}_{\mathbf{I}}} \mathbf{I}_u^t \cdot \mathrm{P}[\mathbf{I}_u^t | \mathbf{\hat{I}}_u^t]
\\&= \sum_{\mathbf{I}_u^t \in \mathbf{S}_{\mathbf{I}}} \mathbf{I}_u^t \cdot \mathbf{\hat{I}}_u^t \cdot \mathbf{\bar{a}}_u^t + \sum_{\mathbf{I}_u^t \in \mathbf{S}_{\mathbf{I}}} \mathbf{I}_u^t \cdot \mathbf{g}_{u} \cdot (1-\mathbf{\bar{a}}_u^t)
\end{split}
\end{equation}
For each position $f$ of this expectation, we can write
\begin{equation}
\begin{split}
\mathbb{E}[i_{u,f}^t] &= \mathbb{E}[\mathbf{I}_u^t | \mathbf{\hat{I}}_u^t]_f 
\\&= \sum_{\mathbf{I}_u^t \in \mathbf{S}_{\mathbf{I}}} i_{u,f}^t \hat{i}_{u,f}^t \bar{a}_{u,f}^{t} + \sum_{\mathbf{I}_u^t \in \mathbf{S}_{\mathbf{I}}} i_{u,f}^t g_{u,f} (1-\bar{a}_{u,f}^{t})
\\&= \hat{i}_{u,f}^t \bar{a}_{u,f}^{t} + g_{u,f} (1-\bar{a}_{u,f}^{t})
\end{split}
\end{equation}
Since we use $a_{u,f}^t$ to estimate probability $\bar{a}_{u,f}^{t}$, we have
\begin{equation}
\begin{split}
\mathbb{E}[i_{u,f}^t] 
= \hat{i}_{u,f}^t a_{u,f}^t + g_{u,f} (1-a_{u,f}^t)
\end{split}
\end{equation}
\end{proof}

Plugging (\ref{eq:estimation_true}) into (\ref{eq:expectedRevEqn0}), we get the expected achievable revenue as
\begin{align}
\label{eq:expRtotal}
\mathbb{E}[R_{\mathrm{total}}^\tau]&=\sum\limits_{k=0}^{K-1} \gamma^k \sum\limits_{t = n(\tau+k)}^{n(\tau+k+1)-1} \sum\limits_{f =0}^{F-1} \sum\limits_{u=0}^{U-1} \left(\hat{i}_{u,f}^t a_{u,f}^t + g_{u,f} (1-a_{u,f}^t) \right) \nonumber \\ 
&  \cdot ( \beta - c_{\mathrm{bs-ue}} 
- (1-d_f^{\tau+k})c_{\mathrm{cl-bs}}) \nonumber \\
&-\sum_{k=0}^{K-1} \gamma^k \cdot c_{\mathrm{plc}} \sum_{f=0}^{F-1} d_f^{\tau+k} (1-d_f^{\tau+k-1}),
\end{align}

Note that this expected achievable revenue in (\ref{eq:expRtotal}) only changes the objective function in (\ref{optproblem}).
Therefore, the optimization problem is still a multistage Knapsack problem.

\subsection{Problem Transformation}
\label{sec:z_transformation}


While a multi-stage Knapsack problem is typically solvable using \ac{ilp}, the multiplication of two variables are not allowed.
Since we have a product term $d_f^{\tau+k} (1 - d_f^{\tau+k-1})=d_f^{\tau+k} - d_f^{\tau+k} d_f^{\tau+k-1}$ in (\ref{eq:expRtotal}), we introduce a new variable $z_f^{\tau+k}$ to replace this term as
\begin{equation}
\label{z_variable}
z_f^{\tau+k} \coloneqq d_f^{\tau+k} \cdot d_f^{\tau+k-1}.
\end{equation}
Then, we enforce the following constraints to replace the multiplicative term.
\begin{align}
&C3:z_f^{\tau+k} \leq d_f^{\tau+k}, \\& C4: z_f^{\tau+k} \leq d_f^{\tau+k-1},
\\& C5: z_f^{\tau+k} \geq d_f^{\tau+k} + d_f^{\tau+k-1}-1,
\\ &C6:z_f^{\tau+k}\in\{0,1\}, \quad \forall f \in [0,F-1], k \in [0,K-1],
\end{align}
where $C3$ ensures that if $d_f^{\tau+k}=0$, then $z_f^{\tau+k}$ must be $0$. 
Similarly, $C4$ ensures that if $d_f^{\tau+k-1}=0$, then $z_f^{\tau+k}$ must be $0$. 
Besides, constraint $C5$ enforces that $z_f^{\tau+k}=1$ only if $d_f^{\tau+k} = d_f^{\tau+k-1} = 1$. 
Finally, $C6$ defines $z_f^{\tau+k}$ as a binary variable.

Therefore, we transform the original problem as 
\begin{subequations}
\label{optproblem_Transformed}
\begin{align}
\underset{\{\{d_f^{\tau+k}\}_{f=0}^{F-1}\}_{k=0}^{K-1}} {\mathrm{maximize}} & \qquad  \mathbb{E}[\tilde{R}_{\mathrm{total}}] , \tag{\ref{optproblem_Transformed}} \\
\mathrm{subject ~to} &\qquad C1, ~C2, ~C3, ~C4, ~C5, ~C6,
\end{align}
\end{subequations}
where $\mathbb{E}[\tilde{R}_{\mathrm{total}}] = \sum_{k=0}^{K-1} \gamma^k \sum_{t = n(\tau+k)}^{n(\tau+k+1)-1} \sum_{f=0}^{F-1} \sum_{u=0}^{U-1} \big( \hat{i}_{u,f}^t a_{u,f}^t + g_{u,f} (1-a_{u,f}^t)\big) \cdot ( \beta - c_{\mathrm{bs-ue}} 
- (1-d_f^{\tau+k})c_{\mathrm{cl-bs}}) - \sum_{k=0}^{K-1} \gamma^k \cdot c_{\mathrm{plc}} \sum_{f=0}^{F-1} (d_f^{\tau+k}-z_f^{\tau+k})$.

It is easy to check \bblue{that} the objective function and \bblue{the} constraints are all linear, and the variables are integer.
Therefore, (\ref{optproblem_Transformed}) is a \ac{ilp} problem. While it remains NP-hard, this problem is solvable using existing solvers such as Gurobi \cite{gurobi}. 
The worst-case scenario for solving this problem is employing a brute-force search, which exhaustively explores all possible combinations of the decision variables $\{\{d_f^\tau\}_{f=0}^{F-1}\}_{\tau=0}^{K-1}$. 
Given that there are $F$ files and $K$ cache placement slots, each of the $KF$ file-slot pairs independently allows two choices (store or not store). 
Thus, the \bblue{worst case} computational \bblue{time} complexity is $\mathcal{O}(2^{KF})$.


\subsection{Cache Placement Procedure}
\label{sec:caching_proc}
After receiving the trained model, each user predicts the content requests for the next $K$ cache placement slots. 
These predicted results are then used to estimate the actual content demand, which is subsequently uploaded to the \ac{es}. 
Based on this information, the \ac{es} computes the optimal cache placement decisions by solving (\ref{optproblem_Transformed}). 
Now, given the cache decisions $\{\{d_f^{{\tau+k}^{*}}\}_{f=0}^{F-1}\}_{k=0}^{K-1}$ for the current cache placement slot $\tau$ to future $\tau+K-1$ cache placement slots, we store the files into the cache at the current slot $\tau$ based on $\{d_f^{{\tau}^*}\}_{f=0}^{F-1}$. 
Then, when we move to the next caching slot $\tau+1$, we have new actual historical content requests information that are leveraged to predict the demands for the next $nK$ mini-slots as described in Section \ref{sec:predion_proc}. 
These predicted demands are then utilized to get the cache placement decisions for $(\tau+1)$ to $(\tau+K)$ caching slots by solving (\ref{optproblem_Transformed}). 
This procedure is repeated in every $\tau$ intervals to utilize newly observed historical data for making future demand predictions, which then guides the caching policy optimization. 
The entire caching procedure after \ac{fl} trained Transformer solution $\Theta^*$ is shown as Algorithm \ref{algo1}.

\begin{algorithm}[!t]
\caption{Cache Placement Methods}
\label{algo1}
\begin{algorithmic}[1]
\STATE \textbf{Input:} Global model $\Theta^*$, total cache placement slots $\mathrm{T}$, number of future caching slots $K$
\STATE Each \ac{ue} receives the trained global model $\Theta^*$ from the server
\FOR{$\tau = 0, 1, 2, \dots, \mathrm{T}-1$}
    \FOR{$u=0, 1,2,...,U-1 $}
        \STATE Predicts future content requests $\{\{ \hat{i}_{u,f}^{t} \}_{f=0}^{F-1}\}_{t =n\tau}^{n(\tau+K)-1}$ based on available historical information using $\Theta^*$
        \FOR{all mini-slots $t = n\tau,n\tau+1, ..., n(\tau+K)-1$}
            \STATE Estimates the actual content requests as $i^{t}_{u,f,\mathrm{est}} = \hat{i}_{u,f}^t a_{u,f}^t + g_{u,f} (1-a_{u,f}^t)$
        \ENDFOR
        \STATE Uploads $\{\{i^{t}_{u,f,\mathrm{est}} \}_{f=0}^{F-1}\}_{t =n\tau}^{n(\tau+K)-1}$ to the \ac{es}
    \ENDFOR
    \STATE \Ac{es} uses all estimated requests $i_{u,f,\mathrm{est}}^t$ from all users and solve (\ref{optproblem_Transformed}) using existing solver (e.g., Gurobi) to get optimal cache decisions $\{\{d_f^{\tau+k^*} \}_{f=0}^{F-1}\}_{k=0}^{K-1}$
    \STATE \Ac{es} store files into its cache in the current caching slot $\tau$ based on $\{d_f^{\tau^*} \}_{f=0}^{F-1}$ 
    \STATE Users place content requests and the \ac{cs} delivers the requested content from the \ac{es} if its in the \ac{es}' cache, otherwise it extracts the content from its cloud and deliver to the requester using the \ac{isp}
\ENDFOR
\STATE \textbf{Output:} Optimal cache decisions $\{\{d_f^{\tau+k^*} \}_{f=0}^{F-1}\}_{k=0}^{K-1}$, system cache-hit-ratio and revenue
\end{algorithmic}
\end{algorithm}
It is worth noting that since users only share their estimation results, their true future requests are not revealed.
Besides, our above solution is general, such that any \ac{ml} can be used easily for demand prediction. 
Suppose the time complexity of performing the forward propagation of the trained $\Theta^*$ is  $\mathcal{O}(NF \times nKF))$. 
There are $UF(nK-1)$ iterations for cache decision computation at each $\tau$, and a total $\mathrm{T}$ number of caching slots. 
Therefore, the overall computational time complexity of this algorithm\footnote{We ignored the time complexity of computing and uploading $i_{u,f,\mathrm{est}}^t$ for simplicity.} is $\mathcal{O} \left(\mathrm{T} \times \left[(NF \times nKF) \times (UF(nK-1)) + 2^{KF} \right] \right)$. 

\section{Simulation Results and Discussions}
\label{sec:sim_results}
This section presents performance evaluation metrics, our simulation settings, and result comparisons with existing baselines.

\subsection{Evaluation metric}
While we use Algorithm \ref{algo1} to find the cache placement decisions sequentially, we want to calculate the actual expected revenue of the video caching network as described in Section \ref{sec:rev_opt}.
Given that we cached the content based on $\{d_f^{\tau^*} \}_{f=0}^{F-1}$, 
we calculate the instantaneous revenue for these cached content after the end of the current cache placement slot $\tau$ as
\begin{align}
\label{eq:R_total}
R_{\mathrm{eva}}^{\tau}
\coloneqq & \sum\limits_{t=n\tau}^{n(\tau+1)-1}\sum\limits_{f=0}^{F-1} \sum\limits_{u =0}^{U-1} i_{u,f}^t ( \beta - c_{\mathrm{bs-ue}} - (1-d_f^{\tau^*}) c_{\mathrm{cl-bs}}) \nonumber\\
&- c_{\mathrm{plc}} \sum_{f=0}^{F-1} d_f^{\tau^*} (1-d_f^{\tau-1}).
\end{align}
Note that since $R_{\mathrm{eva}}^{\tau}$ is calculated after the end of the $[n\tau, n(\tau+1)-1]$ mini-slots, (\ref{eq:R_total}) represents the instantaneous revenue obtained from storing $\{d_f^{\tau^*}\}_{f=0}^{F-1}$ with the ground truth content request information.
Therefore, we stress that while we used the expected content requests over $K$ caching slots to design our caching policy in Section \ref{sec:two_stage_solution}, we are still evaluating the \emph{actual} system revenue obtained from our proposed caching strategy.

The instantaneous revenue largely depends on (a) which files are in the cache and (b) variations in users' content requests. 
User request patterns vary across different cache placement slots; for instance, in certain slots, users may collectively request a limited subset of files, whereas in other slots, each user may request distinct files. 
Such variability may inherently yield fluctuations in system revenue. 
We observe similar trends in our simulation results. 
In Fig. \ref{fig:taus}, we compare the instantaneous revenues of our proposed two-stage solution with the \emph{ground truth case}.
\bblue{The simulation assumptions and other key parameters used to obtain Fig. \ref{fig:taus} and the rest of the results are discussed in the sequel.}
The ground truth case essentially is the special case in which a Genie has perfect knowledge of all future content requests $i_{u,f}^t$ for all mini-slots in the subsequent $K$ cache placement slots, which are then used in the objective function in  (\ref{eq:rev_function}). 
We then follow the transformation steps outlined in Section \ref{sec:z_transformation}. 
Consequently, the resulting optimization problem under \emph{perfect} future knowledge is expressed as
\begin{subequations}
\label{GT_optproblem_Transformed}
\begin{align}
\underset{\{\{d_f^{\tau+k}\}_{f=0}^{F-1}\}_{k=0}^{K-1}} {\mathrm{maximize}} & \qquad  R_{\mathrm{total}}^{\tau} , \tag{\ref{GT_optproblem_Transformed}} \\
\mathrm{subject ~to} &\qquad C1, ~C2, ~C3, ~C4, ~C5, ~C6,
\end{align}
\end{subequations}
where $R_{\mathrm{total}}^\tau 
    =\sum_{k=0}^{K+\tilde{K}-1} \gamma^k \big[\sum_{t = n(\tau+k)}^{n(\tau+k+1)-1} \sum_{f =0}^{F-1} \sum_{u =0}^{U-1} i_{u,f}^t( \beta - c_{\mathrm{bs-ue}}  - 
    (1-d_f^{\tau+k})c_{\mathrm{cl-bs}}) \big] - \sum_{k=0}^{K+\tilde{K}-1} \gamma^k \cdot c_{\mathrm{plc}} \sum_{f=0}^{F-1} (d_f^{\tau+k} -z_f^{\tau+k})$.

As shown in Fig. \ref{fig:taus}, there are small performance differences between our proposed two-stage solution and the ground truth best case. 
These differences stem from the inherent imperfections in future demand predictions, as perfect forecasting is difficult.
Nevertheless, due to the relatively high accuracy of our Transformer-based predictions and the incorporation of content popularity in cache decision-making, the performance of our two-stage solution remains closely aligned with that of the ground truth.

Now, since we cannot practically consider $\mathrm{T} \to \infty$ to evaluate (\ref{eq:total_revenue}), we use (\ref{eq:R_total}) to calculate the expected system revenue over a finite number of caching slots as
\begin{align}
    R_{\mathrm{eva}} = \frac{1}{\mathrm{T}} \sum_{\tau = 0}^{\mathrm{T}-1}R_{\mathrm{eva}}^{\tau}. 
\end{align}

\begin{figure}
\centering
\includegraphics[scale = 0.5]{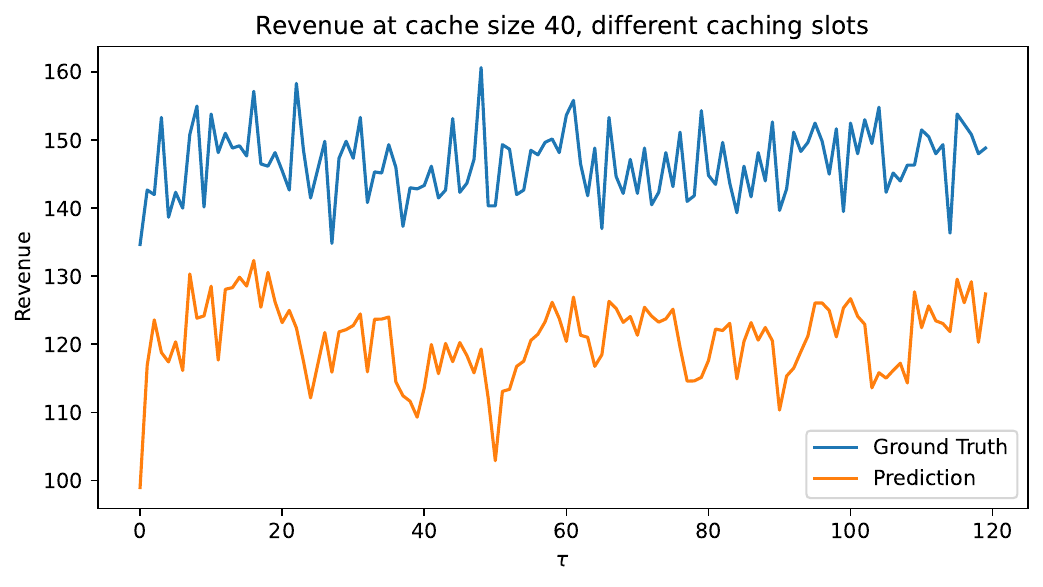}
\caption{System revenue at different cache placement slots}
\label{fig:taus}
\end{figure}

\subsection{Simulation Assumptions and Parameters} 
Since to the best of our knowledge there is no publicly available real world dataset with spatial and temporal user request information, we consider the following synthetic content request model\footnote{Other request models can also be easily incorporated into our proposed solution} to generate users' datasets. 

\subsubsection{Content Request Model} 
We assume that the user's request model follows a day-by-day pattern. 
There are two request stages each day: (a) a random requests initialization stage and (b) a subsequent requests generation stage, in which the following requests are made based on the previous information. 
We assume that each user selects a genre $\mathfrak{g}$ based on their genre preference $p_{u,\mathfrak{g}}$ at the beginning of a day and requests content from only the selected genre that whole day. 
During the first stage, the user will randomly request $L$ distinct content based on content popularity $P_{\mathfrak{g},f}$, similarity\footnote{We consider each content's distinctive feature set that can be used to calculate the cosine similarity of the content within the same genre.} and time-dependent features. Specifically, similarity is jointly considered with content popularity to compute a refined popularity score, which guides the random selection of content.
Then, we generate the following $M$ requests at the second stage based on these past $L$ requests. 
Denote the initial $L$ requested content by $\mathcal{F}_{u,L} = \{ f_{0}, f_{1},...,f_{l},...,f_{L-1}\}$ and the next $M$ requested content by $\mathcal{F}_{u,M}= \{ f''_0, f''_1,...,f''_m,...,f''_{M-1}\}$. 
Let us also denote the feature set of the $f_l^{\mathrm{th}}$ content by $\pmb{\phi}_{f_l}$. 
We then consider time dependency when generating the next requests by introducing a forgetting factor $w_l$ that makes the latest requests contribute more than the past requests. 
More specifically, we define  $w_l$ as
\begin{equation}
w_l \coloneqq e^{-\frac{L-l+1}{b}}, \ l=0,1,...,L-1
\end{equation}
where $b$ is a constant hyper-parameter.

Now, we find the most similar files to the previous $L$ files.
Particularly, we incorporate the effect of temporally weighted historical requests to compute a similarity score as follows
\begin{equation}
S_{u,f''} = \sum \limits_{l = 0}^{L-1} w_l \cdot \frac{\pmb{\phi}_{f_l} \cdot \pmb{\phi}_{f''} }{\Vert \pmb{\phi}_{f_l} \Vert \ \Vert \pmb{\phi}_{f''}\Vert}
\end{equation}
where $f''$ denote to be requested content and $f'' \in \mathcal{F} \setminus \mathcal{F}_{u,L}$. 

Then, we compute a new score based on similarity and content popularity inside genre as
\begin{equation}
\label{eq:score}
\bar{S}_{u,f''} = \lambda \cdot \frac{\mathrm{exp}(S_{u,f''})}{\sum\limits_{f'' \in \mathcal{F} \setminus \mathcal{F}_{u,L}} \mathrm{exp}(S_{u,f''})} + (1-\lambda) \cdot \frac{\mathrm{exp}(P_{\mathfrak{g},f''})}{\sum\limits_{f'' \in \mathcal{F} \setminus \mathcal{F}_{u,L}} \mathrm{exp}(P_{\mathfrak{g},f''})}
\end{equation}
where $\lambda \in (0,1)$ is a weighting factor. 
Then, we assume the user selects the Top-$M$ files based on $\bar{S}_{u,f''}$ as the content requests for the next $M$ slots. 
This process is repeated until the user makes $Q$ requests per day. 
Similarly, the process is repeated for $E$ days.
Algorithm \ref{algo2} summarized the content request model.

\begin{algorithm}[!t]
\caption{User Request Model}
\label{algo2}
\begin{algorithmic}[1]
\STATE \textbf{Input:} Total number of days $E$, user request per day $Q$, previous number of requests $L$, following number of requests $M$, genre preference $p_{u,\mathfrak{g}}$, content popularity $P_{\mathfrak{g},f}$
\FOR{$u=0, 1,2,..,U-1$}
    \FOR{$\epsilon = 0, 1,2,...,E-1$}
        \STATE Select a genre $\mathfrak{g}$ based on genre preference $p_{u,\mathfrak{g}}$ of user $u$, the following content selection is within genre $\mathfrak{g}$
        \FOR{$\mu=\epsilon Q+1,\epsilon Q+2,...,(\epsilon +1) Q$}
            \STATE $\mu' \gets \mu-\epsilon Q$
            \IF{$\mu'\leq L$}
                \STATE $\ell \gets \mu'$
                \STATE Select content randomly based on popularity $P_{\mathfrak{g},f}$, similarity and time-dependent features
            \ELSE
                \IF{$\mu'= L + iM+1, \ i = 0,1,2,\dots$}
                    \STATE Select previously requested content from $\mu'-L$ to $\mu'-1$ as $\mathcal{F}_{u,L}$
                    \STATE Compute score $\bar{S}_{u,f''}$ of to be requested content $f'' \in \mathcal{F} \setminus \mathcal{F}_{u,L}$ using (\ref{eq:score})
                    \STATE Select Top-$M$ files based on $\bar{S}_{u,f''}$ as $\mathcal{F}_{u,M}$
                \ENDIF
            \ENDIF
        \ENDFOR
    \ENDFOR
\ENDFOR
\STATE \textbf{Output:} Total content requests
\end{algorithmic}
\end{algorithm}

\subsubsection{Simulation Parameters} 
For our simulation, we consider $U=50$ users and $F=240$ files\footnote{Due to computational resource constraints, we consider a relatively smaller number of users and files than those typically encountered in real-world scenarios (e.g., Netflix, which hosts millions of video contents). However, the underlying concept of revenue optimization remains unchanged.}.
Besides, we assume that each file has a size $B = 1$ and belongs to one of the $G=3$ genres.
Content popularity in genre $P_{\mathfrak{g},f}$ follows a Zipf distribution with exponent $\tilde{\gamma}=1.2$. 
For the genre preference $p_{u,\mathfrak{g}}$ we use a symmetric Dirichlet distribution, $\mathrm{Dir}(\alpha_u)$, where $\alpha_u = 0.3$ is the concentration parameter. 
Furthermore, $L=7$, $M=5$, $b=0.5$, and $\lambda=0.5$ are used to generate the content requests.  
We generate $Q = 107$ requests daily and use $80$ days per user\footnote{While larger values could improve accuracy, these settings offer a practical trade-off between model performance and computational efficiency under limited resources.} as training data.

\subsubsection{Transformer training setting}

Let $t_m$ denote the first mini-slot in the label of the $m^{\mathrm{th}}$ training sample. 
The corresponding feature and label for this sample are given by $\mathbf{x}_u^m = (\mathbf{I}_{u}^{t_m-N}, \mathbf{I}_{u}^{t_m-N+1}, \cdots, \mathbf{I}_{u}^{t_m-1}) \in \mathbb{R}^{N\times F}$ and $\mathbf{y}_u^m = (\mathbf{I}_{u}^{t_m}, \mathbf{I}_{u}^{t_m+1}, \cdots, \mathbf{I}_{u}^{t_m+nK-1}) \in \mathbb{R}^{nK\times F}$, respectively. Each user generates training data using a sliding window approach with a step size of $n$. 

We use a vanilla Transformer \cite{attentionAllyouNeed} as our \ac{ml} model. 
In our implementation, the Transformer has $6$ encoder/decoder layers with $2$ heads. 
The input data will first go through an embedding layer of dimension $512$. 
The feed-forward \cite{attentionAllyouNeed} dimension is $1024$.
Besides, we use $\kappa=5$, $r_g=500$ global rounds, the cross-entropy loss, and SGD optimizer with a learning rate of $0.15$ for the \ac{fedavg} \cite{McMahanFL} algorithm. 
Moreover, $\beta= 3$, $c_{\mathrm{bs-ue}} = 0.5$, $c_{\mathrm{cl-bs}} = 2$, $c_{\mathrm{plc}} = 1.5, \gamma = 0.8$, $S \in [0,240]$, $n=2$,\bblue{and $K=5$.
The prediction accuracies in $n \times K = 5 \times 2=10$ mini-slots are shown in Table \ref{tab:predAccWtransformer}, which shows that the accuracy drops significantly as we predict the request farthest from the current slot.
Since the prediction accuracy at the last mini-slot of our chosen $\tau+K-1$ is so low that $d^{\tau+K-1}$ is dominated by the global popularity distribution, the revenue for times $\ge \tau+K$ does not play a significant role.
Therefore,we set $\tilde{K}=0$ for simplicity.}

\begin{table*}[t]
\centering
\caption{Prediction Accuracy in Different Mini-Slots with Vanilla Transformer \cite{attentionAllyouNeed} (in $3$ Independent Trials)}
\fontsize{6}{8}\selectfont
\begin{tabular}{|c|c|c|c|c|c|c|c|c|c|} \hline
\textbf{Slot $0$} &
\textbf{Slot $1$} & \textbf{Slot $2$} & \textbf{Slot $3$} &
\textbf{Slot $4$} &
\textbf{Slot $5$} &
\textbf{Slot $6$} &
\textbf{Slot $7$} &
\textbf{Slot $8$} &
\textbf{Slot $9$}\\ \hline
$0.8323 \pm 0.0091$ &  $0.8055 \pm 0.0122$ & $0.7914 \pm 0.0155$ & $0.7731 \pm 0.0175$ & $0.7568 \pm 0.0147$ & $0.7450 \pm 0.0149$ & $0.7352 \pm 0.0156$ & $0.7221 \pm 0.0137$ & $0.7031 \pm 0.0164$ & $0.6740 \pm 0.0198$ \\ 
 \hline
\end{tabular}
\label{tab:predAccWtransformer}
\end{table*}

\subsection{Performance Analysis and Comparisons}
In this work, we use the widely popular \emph{Transformer} \cite{attentionAllyouNeed} as our \ac{ml} model since it yielded superior performance over other popular models like \ac{lstm} and \ac{gru} in our simulation. 
\begin{figure}[t]
\centering
\includegraphics[scale = 0.4]{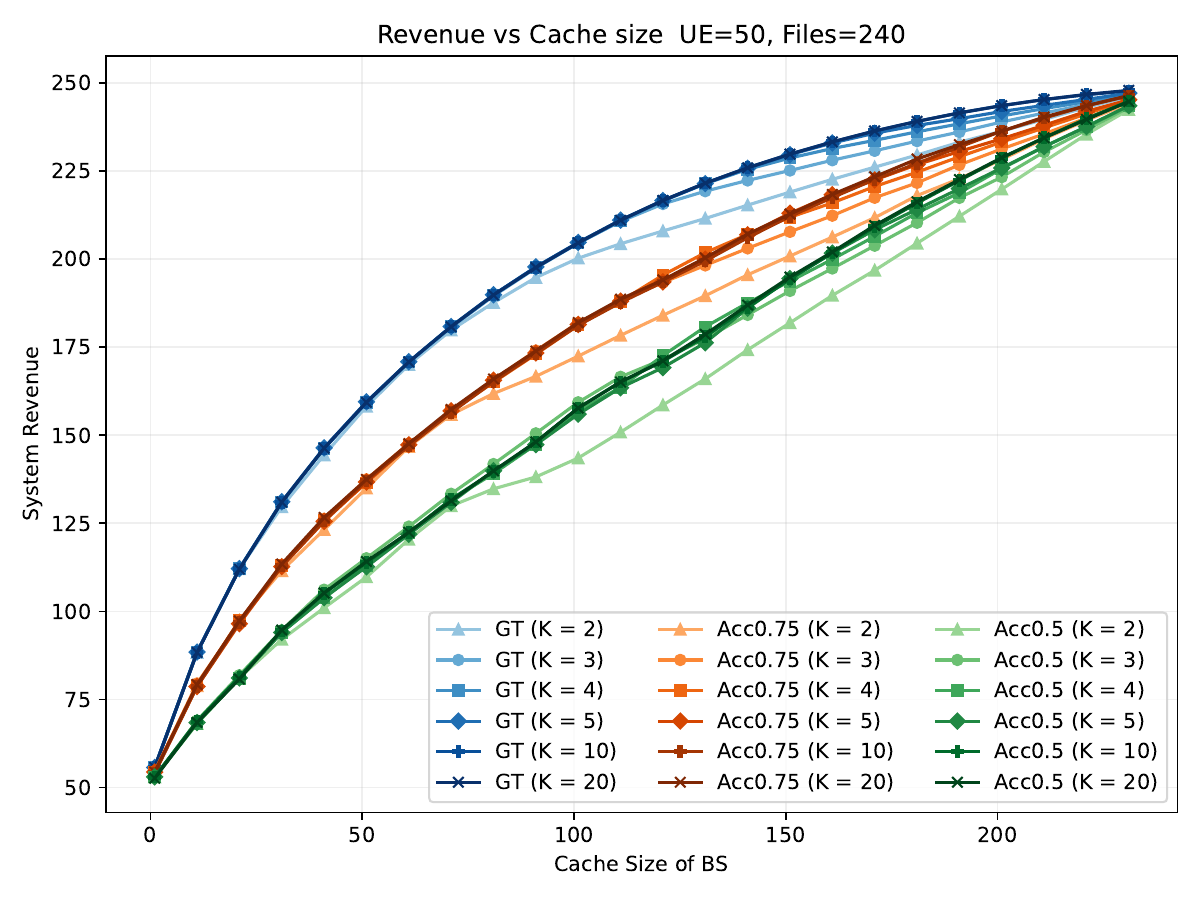}
\caption{Revenue comparison for different number of caching slots under different prediction quality}
\label{fig:KcompareRev}
\end{figure}

\subsubsection{Performance Analysis}
Now, we want to investigate the impact of different parameters on revenue and find the optimum ones. 
From (\ref{eq:rev_function}), it is clear that the parameters that affect system revenue are number of future slots $K$, discount factor $\gamma$, and cache placement cost $c_{\mathrm{plc}}$ (since $c_{\mathrm{plc}}$ and $c_{\mathrm{cl-bs}}$ jointly affect revenue, we can fix one and vary the other). 
The parameter comparison should consider different prediction qualities since the prediction accuracy also affects the system revenue.

Since we consider multi-slot revenue optimization, how many slots we want to look at in the future is essential.
Intuitively, a longer window, i.e., a large value of $K$, will yield better system revenue. 
This is because larger $K$ allows us to consider keeping files that remain popular over a longer period, reducing the relative cache replacement cost.

However, a large value of $K$ can be computationally challenging since we need to repeatedly solve a larger-dimensional multi-stage Knapsack problem.
Therefore, we want to find a proper $K$ that is sufficient to find quasi-optimal cache decisions. 
We compare the effect of different $K$ for different prediction qualities\footnote{For this comparison we use a "genie-plus-error" model, in which ground-truth requests are randomly replaced with erroneous requests to simulate different prediction accuracies.} case for parameter selection, shown as Fig. \ref{fig:KcompareRev}. 
As expected, the results show that the revenue increases as $K$ increases.
Besides, when prediction accuracy decreases, the revenue also decreases since prediction errors mislead the \ac{bs} to store files that will not be requested in the future. 
The results show that the revenue improvement is less noticeable beyond $K=5$ future slots.
Therefore, we consider $K=5$ for the rest of the simulation results.

We then want to find the optimal setting for $\gamma$ and $c_{\mathrm{plc}}$.
When we reduce the decaying factor $\gamma$, future requests will contribute less to the current decision. 
In other words, a small $\gamma$ leads to a myopic view that puts less weight on the farthest future requests, 
leading to the possibility of more content exchange in the future.
Similarly, as $c_{\mathrm{plc}}$ decreases, storing new files in the cache \bblue{becomes cheaper}. The system tends to cache more content that fulfills the current requests, as bringing new content to the storage is cheaper. Thus, we expect reducing $\gamma$ and $c_{\mathrm{plc}}$ to result in similar outcomes; both cases lead to decreased revenue - though note that $c_{\mathrm{plc}}$ is a system parameter reflecting the cost of cache placement, while $\gamma$ is a hyperparameter of the algorithm. 
However, when we increase the cache size, the impact of reduced $\gamma$ and $c_{\mathrm{plc}}$ may not be obvious since the \ac{es} has more room to prefetch more contents that will be requested in the future. 
We compare the effect of different $\gamma$ and $c_{\mathrm{plc}}$ in different prediction quality cases for parameters in Fig. \ref{fig:GammacompareRev} and Fig. \ref{fig:c_clbs_compareRev}.
These results follow the above discussions. 
Moreover, the revenue decreases as prediction accuracy decreases. Specifically, we observe that when $\gamma>0.7$ and $c_{\mathrm{plc}}>1.2$, the revenue is close to the optimal revenue.
Thus, we set $\gamma=0.8$ and $c_{\mathrm{plc}}=1.5$ for our simulations.

\begin{figure}
\centering
\includegraphics[scale = 0.4]{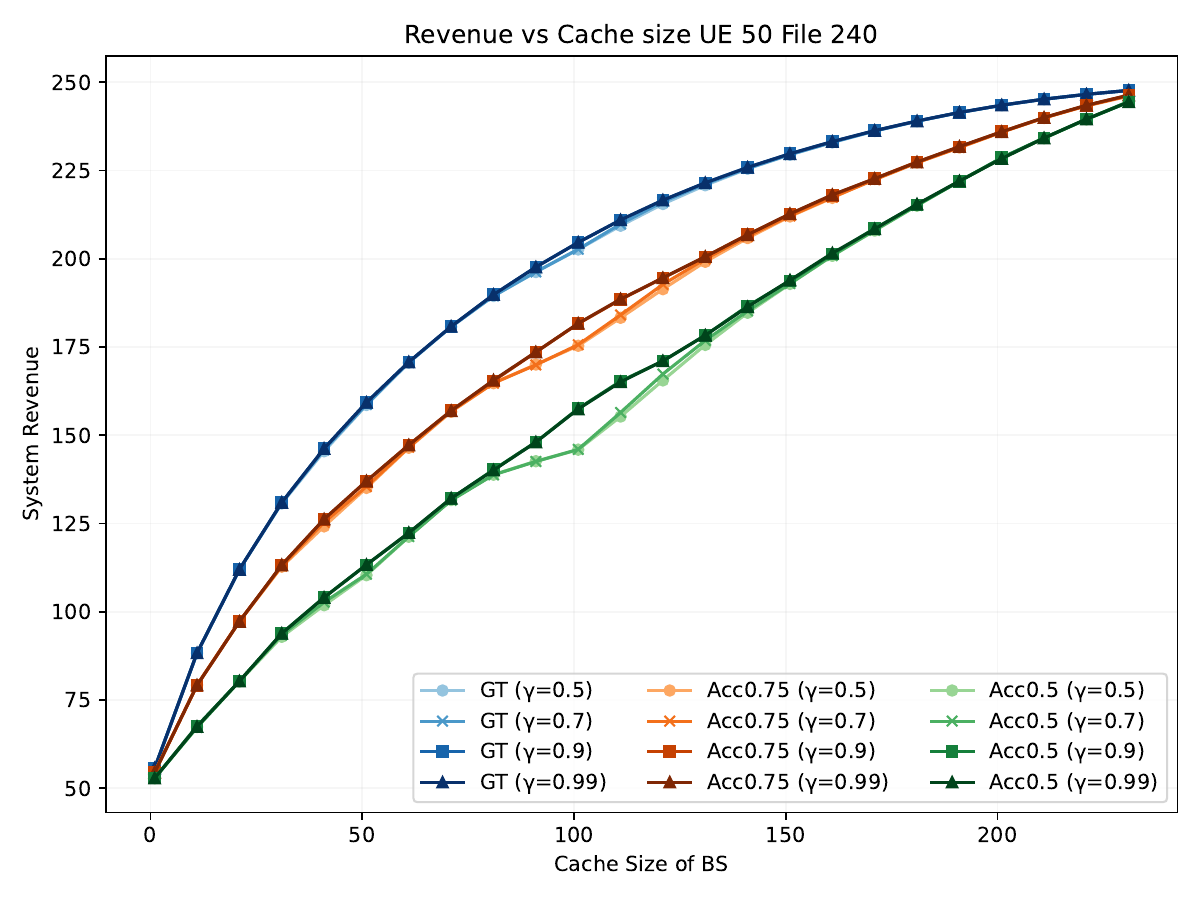}
\caption{Revenue comparison for different $\gamma$ under different prediction quality, $10$ caching slots}
\label{fig:GammacompareRev}
\end{figure}

\begin{figure}
\centering
\includegraphics[scale = 0.4]{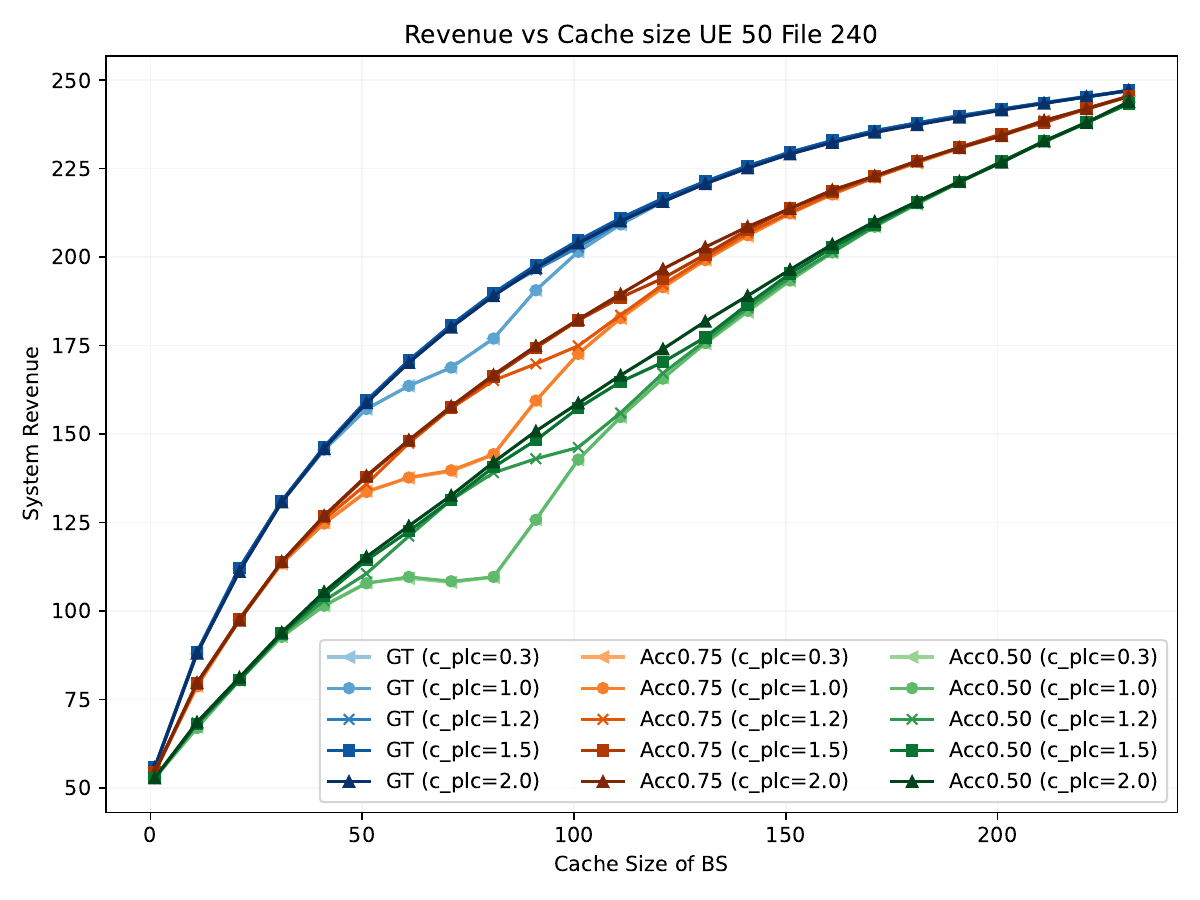}
\caption{Revenue comparison for different $c_{\mathrm{plc}}$ under different prediction quality, $5$ caching slots}
\label{fig:c_clbs_compareRev}
\end{figure}

\subsubsection{Baseline Comparisons} 
To evaluate the performance of our proposed method, we consider several baselines. 

\begin{itemize}
\item Centralized SGD (\textbf{C-SGD}): This baseline assumes a \emph{Genie} has access to all users' datasets on which the \ac{ml} model $\Theta$ is trained. 
Upon getting the trained model, we use it to predict user requests $\hat{i}_{u,f}^t$. The subsequent procedures follow those outlined in Sections \ref{sec:request_estimation} to \ref{sec:caching_proc}. 

\item One slot revenue optimization caching (\textbf{One-slot}) \cite{zhang2025revenue}: This baseline, which we developed in our conference version \cite{zhang2025revenue}, only predicts one placement slot into the future and assumes further cache decisions are fixed based on global popularity.
\bblue{Specifically, we train a Transformer to predict the content demands only in the next $n$ mini-slots using the strategy described in \cite{zhang2025revenue}.}
We assume the caching decisions after the next slot, say $\{\{d_f^{\tau+k}\}_{f=0}^{F-1}\}_{k=1}^{K-1}$, are fixed and determined based on global content popularity $G_f$\footnote{We assume the \ac{csp} knows global content popularity information and makes it available to facilitate caching policies.}. Based on these simplifications, the revenue objective function is represented as
\begin{align}
\label{eq:R_total_oneslot}
&\mathbb{E}[R_{\mathrm{one-slot}}^{\tau} ] 
\coloneqq \mathbb{E}[R_{\mathrm{benefit}}^{\tau} - R_{\mathrm{delivery\_cost}}^{\tau} \nonumber \\
&\quad \quad \quad - R_{\mathrm{placement\_cost}}^{\tau} + \gamma R_{\mathrm{placement\_benefit}}^{\tau+1}], \nonumber\\
&\quad = \underbrace{\sum\limits_{t=n\tau}^{n(\tau+1)-1}\sum\limits_{f=0}^{F-1} \sum\limits_{u=0}^{U-1} (\hat{i}_{u,f}^t a_{u,f}^t + g_{u,f} (1-a_{u,f}^t) )\cdot \beta}_{\mathbb{E}[R_{\mathrm{benefit}}^\tau]} \nonumber\\
&\quad-\underbrace{\sum\limits_{t=n\tau}^{n(\tau+1)-1}\sum\limits_{f=0}^{F-1} \sum\limits_{u=0}^{U-1} (\hat{i}_{u,f}^t a_{u,f}^t + g_{u,f} (1-a_{u,f}^t) )}_{\mathbb{E}[R_{\mathrm{delivery\_cost}}^\tau]} \nonumber
\\
&\quad \cdot \underbrace{(c_{\mathrm{bs-ue}} + (1-d_f^{\tau}) c_{\mathrm{cl-bs}})}_{\mathbb{E}[R_{\mathrm{delivery\_cost}}^\tau]}
\nonumber \\
&\quad - \underbrace{c_{\mathrm{plc}} \sum\limits_{f=0}^{F-1} d_f^{\tau} (1-d_f^{\tau-1})}_{\mathbb{E}[R_{\mathrm{placement\_cost}}^{\tau}]}  + 
\gamma \cdot 
\underbrace{c_{\mathrm{plc}} \sum\limits_{f=0}^{F-1} d_f^{\tau+1} d_f^{\tau}}_{\mathbb{E}[R_{\mathrm{placement\_benefit}}^{\tau+1}]}. 
\end{align}
Since cache decisions after $\tau$ are already known, the objective function only considers requests at $\tau$. 
The placement cost at $\tau+1$ still needs to be considered since there is still re-caching cost from $\tau$ to $\tau+1$. 
Therefore, we solve the following optimization problem for this baseline.
\begin{subequations}
\label{optproblem_singleslot}
\begin{align}
\underset{\{d_f^{\tau}\}_{f=0}^{F-1}} {\mathrm{maximize}} & \qquad  R_{\mathrm{one-slot}}^{\tau} \tag{\ref{optproblem_singleslot}} \\
\mathrm{subject ~to} &\qquad C1:\sum_{ f=0}^{F-1} d_f^{\tau} B \leq S, \\ 
&\qquad C2: d_f^{\tau}\in\{0,1\}, \quad \forall f \in[0,F-1],
\end{align}
\end{subequations}
This problem can be solved as 0-1 Knapsack problem. Since all the files are same size, we can use a greedy algorithm \cite[Chapter 2]{knapsackBook} to solve this problem. 

\item Simple estimation for revenue and cache placement (\textbf{SimpEst}): In this baseline, we rely entirely on Transformer's prediction results. 
More specifically, we estimate the actual content request as  
\begin{equation}
\label{simpleEstEqn}
i^{t}_{u,f,\mathrm{est}} \approx \mathrm{argmax} [\hat{\mathbf{I}}_{u}^t] \cdot a_{u,\mathrm{argmax} [\hat{\mathbf{I}}_{u}^t]}^t.
\end{equation}
We then substitute this estimation into (\ref{eq:expectedRevEqn0}) in place of $\mathbb{E}[i_{u,f}^t]$, forming a new expected revenue. After that, we follow the procedures outlined in Section \ref{sec:z_transformation} to \ref{sec:caching_proc} to obtain the optimal caching decisions.

\item Statistics-based cache placement (\textbf{Statistics}): Content stored at \ac{es}' cache based on historical global content popularity.

\item Least recently used (\textbf{LRU}): The \ac{es} tracks all users' historical requests up to the cache placement slot $\tau-1$.
Content requested during $[n\tau-n, n\tau]$ but absent from the cache will be stored, replacing the least recently used content.

\item Random content caching (\textbf{Random}): This naive baseline stores content randomly for each cache placement slot.
\end{itemize}

We now compare the performance of our proposed two-stage solution with the above baselines, in terms of \ac{chr} and system revenue for different cache sizes. 
As the cache size increases, we expect the CHR and the revenue to increase, which is trivial. 
Besides, the ground truth case is expected to provide the best  performance since perfect information about the future guides the cache placement. 
Given a fixed cache size, our proposed method is expected to be better than other baselines, as our algorithm utilizes the prediction and content popularity to estimate the future content requests. 
Furthermore, the C-SGD baseline is expected to perform better than the \ac{fl} version, since the model is trained on all users' datasets in a centralized fashion ignoring the privacy concerns.  
The One-slot caching method, on the other hand, only considers one future slot, resulting in lower performance than the proposed method.
The SimpEst baseline, on the other hand, decides cache placement only on the Transformer's prediction and, thus, is expected to suffer from the prediction inaccuracies, and thus perform worse as distant-future effects have increasing impact on the revenue. 
Moreover, the other naive baselines (e.g., LRU, Statistics, and Random) will also likely fail to capture the dynamic user demands.
More specifically, the LRU baseline only emphasizes the previous content requests, while future content requests can be significantly different. 
The statistical cache placement baseline only considers global popularities without considering user-specific preferences. 
Finally, the Random caching picks files randomly and does not utilize historical content requests at all.

Our simulation results in Figs. \ref{cachingReqSlots1} - \ref{cachingReqSlots2} follow the above analysis. 
As expected, if perfect future request information is available, multi-slot revenue optimization, i.e., the \emph{ground truth case} shown in (\ref{GT_optproblem_Transformed}), achieves the highest performance compared to other baselines
that do not have access to future knowledge.
Besides, the one-slot baseline with future knowledge, termed as \emph{One-Slot-GT}, also performs worse than the performance we get by optimizing (\ref{GT_optproblem_Transformed}).
Now, when we do not know the future exactly and rely on prediction results, the best baseline is C-SGD since it assumes all users' datasets are centrally available to train the model.
Still, our results suggest that the proposed two-stage \ac{fl} solution is nearly identical to the performance of C-SGD.
Besides, our proposed two-stage solution is even better than the One-Slot-GT case.
It is worth noting that both multi-slot and One-slot approaches (for both ground truth with known future and prediction-based cases) have similar performances when the cache size is small ($<20$). 
This is due to the fact that only the most popular content set in the near future dominates the caching decision with a small cache.
However, as cache size increases, the advantage of considering a longer future window becomes more evident.
In contrast, the heuristic baselines (\ac{lru}, statistical caching, and random) perform significantly worse than the prediction-based methods.
When the cache size is relatively small, the revenues of the \ac{lru} and random baselines decrease because recently requested content is unlikely to reappear soon based on our content request model.
Since \ac{lru} only uses historical information and tends to retain recently requested content, it results in frequent cache replacement when the cache size is small, incurring additional costs.
Random caching stores content arbitrarily, and as the cache size increases, the frequency of content exchanges also increases, leading to higher costs.

\hspace{\fill} \\
\subsubsection{\ac{chr} and Revenue Comparison}
We now investigate CHR and revenue for different user numbers in the system.
We consider user numbers ranging from $20$ to $50$, and comparing \ac{chr} and revenue for the case of known ground truth. The \acp{chr} are calculated based on the cache decisions obtained from Algorithm \ref{algo2}. Fig. \ref{placeholder3} shows 
that both \ac{chr} and revenue increase as the cache size increases, which is expected. 
For a small number of \acp{ue}, the \ac{chr} quickly reaches its maximum value, which is also similar for the system revenue.
This is due to the fact that when the number of requests in each cache placement slot is not too large, a small cache size is sufficient to meet the demands.
As the number of \acp{ue} increases, the total achievable revenue also increases, whereas \ac{chr} does not follow the same trend: 
the \ac{chr} reaches $1$, i.e., the maximum value, while revenue has increasing trends. 
This is because file exchanges still occur between different cache placement slots, even when the cached content meets users' demands within a given caching slot.
Therefore, the revenue is a better reflection of the overall system performance, providing more detailed insights into how many user requests are satisfied and the costs incurred for file reloading. 
The \ac{chr}, on the other hand, is less effective in capturing these details.

\begin{figure}
\centering
\includegraphics[scale = 0.45]{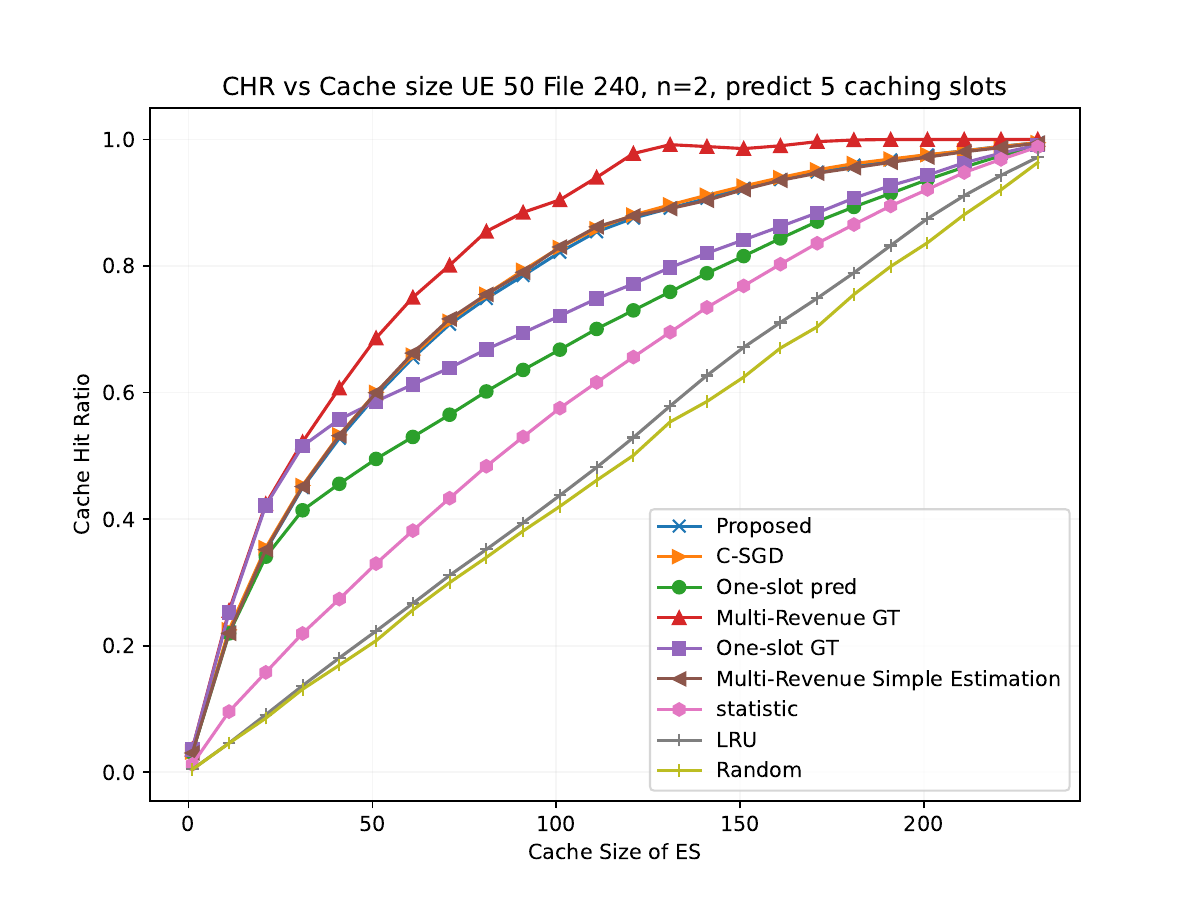}
\caption{\Ac{chr} comparison for different cache placement methods}
\label{cachingReqSlots1}
\end{figure}

\begin{figure}
\centering
\includegraphics[scale = 0.45]{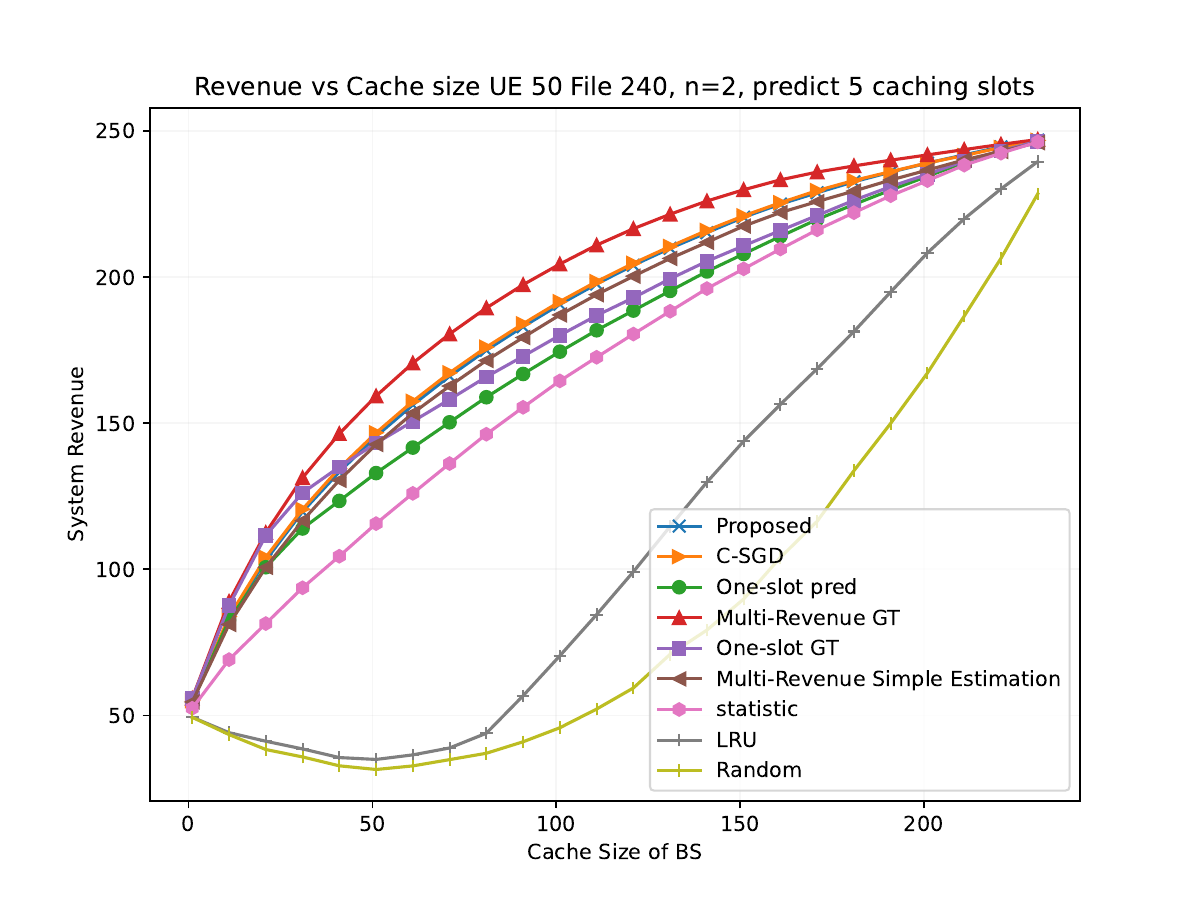}
\caption{Revenue comparison for different cache placement methods}
\label{cachingReqSlots2}
\end{figure}



\begin{figure}
\centering
\includegraphics[scale = 0.4]{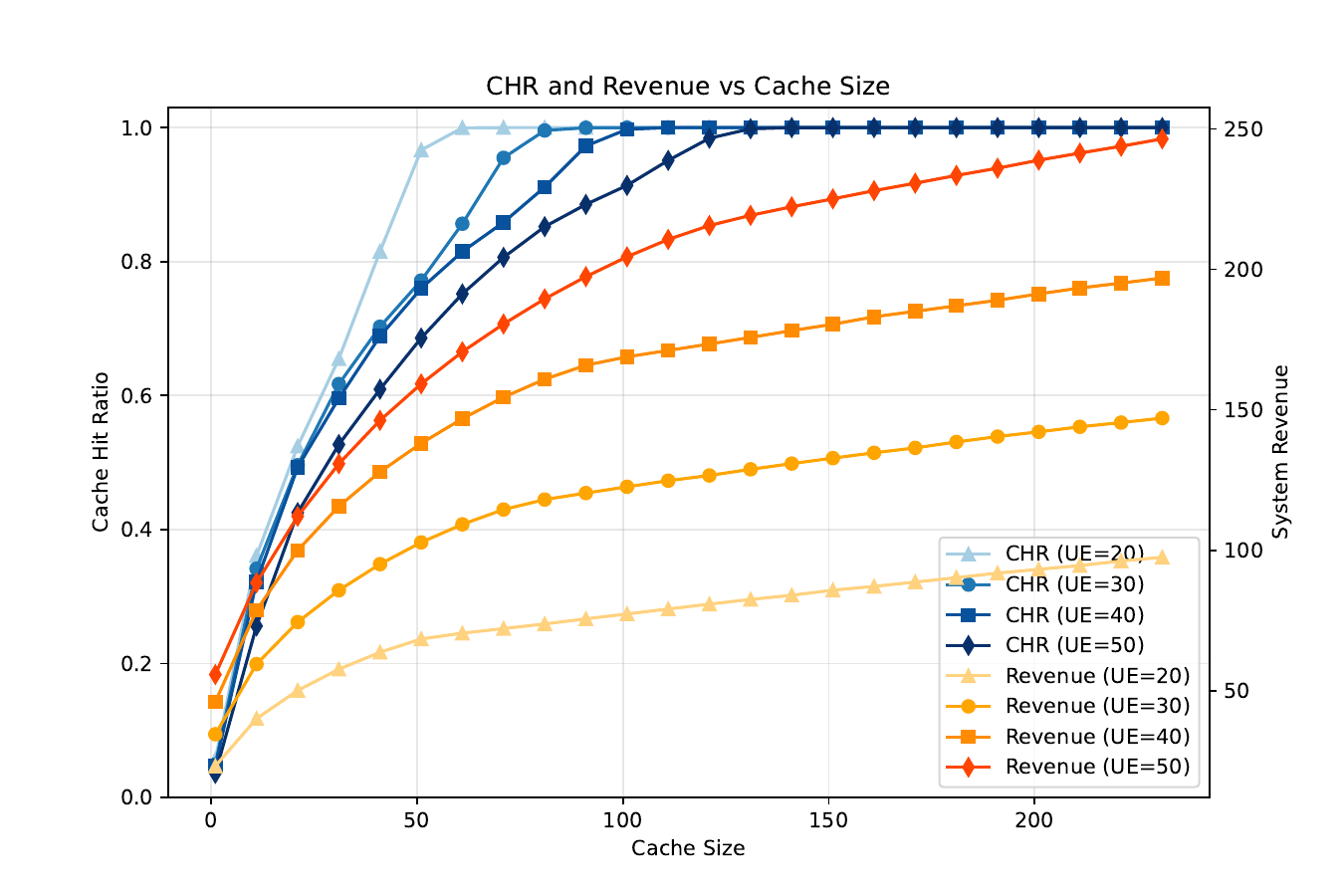}
\caption{CHR and revenue comparison for different number of UEs}
\label{placeholder3}
\end{figure}

\section{Conclusions}
\label{sec:conclusion}
This paper proposed a two-stage solution for wireless video caching networks. 
At first, the multi-step content demands were predicted using a privacy-preserving FL-trained Transformer model.
Then the prediction results were used to estimate users' actual demands, which were then used for a long-term revenue optimization to determine the caching decisions.
Since the original problem is a multi-stage Knapsack problem with an objective function containing multiplications of binary variables, we transformed it by introducing a new variable, which enabled us to use an existing \ac{ilp} solver to solve it efficiently.
The numerical results show our proposed caching method's superiority over other counterparts.
Moreover, our finding suggests that long-term revenue is a more accurate indicator than the widely used \ac{chr} for describing the performance of a caching policy in practical wireless video caching networks.

\bibliographystyle{ieeetr} 
\bibliography{refs} 

\end{document}